\documentclass[a4paper,11pt,twoside]{article}

\usepackage{xcolor}
\usepackage{colortbl}
\definecolor{rank1}{HTML}{00ff00}
\definecolor{rank2}{HTML}{7dff00}
\definecolor{rank3}{HTML}{b1ff00}
\definecolor{rank4}{HTML}{dbff00}
\definecolor{rank5}{HTML}{ffff00}
\definecolor{rank6}{HTML}{ffd000}
\definecolor{rank7}{HTML}{ff9f00}
\definecolor{rank8}{HTML}{ff6700}
\definecolor{rank9}{HTML}{ff0000}

\definecolor{rank06}{HTML}{ff0000}
\definecolor{rank05}{HTML}{ff7b00}
\definecolor{rank04}{HTML}{ffc000}
\definecolor{rank03}{HTML}{ffff00}
\definecolor{rank02}{HTML}{b1ff00}
\definecolor{rank01}{HTML}{00ff00}

\usepackage{multirow}

\usepackage{fontenc}
\usepackage{amsthm,amsmath,amsfonts,amssymb}
\RequirePackage[colorlinks,citecolor=purple,urlcolor=black]{hyperref}%
\RequirePackage{epsfig, graphicx, color}
\RequirePackage{latexsym}
\usepackage{siunitx}

\usepackage{fullpage}
\usepackage{epsf,epsfig}
\usepackage{caption}
\usepackage{subcaption}
\usepackage{natbib}
\usepackage{booktabs}
\usepackage{graphics,graphicx}
\usepackage{enumitem}%
\usepackage{float}

\usepackage[]{algorithm2e}
\LinesNumbered

\usepackage{bbm}

\usepackage{color}
\usepackage{tikz}
\usetikzlibrary{shapes,decorations,arrows,calc,arrows.meta,fit,positioning}
\tikzset{
    -Latex,auto,node distance =1 cm and 1 cm,semithick,
    state/.style ={ellipse, draw, minimum width = 0.7 cm},
    point/.style = {circle, draw, inner sep=0.04cm,fill,node contents={}},
    bidirected/.style={Latex-Latex,dashed},
    el/.style = {inner sep=2pt, align=left, sloped}
}
\usepackage{dsfont}

\theoremstyle{plain}
\newtheorem{theorem}{Theorem}

\newtheorem{prop}[theorem]{Proposition}

\theoremstyle{definition}

\newtheorem{example}{Example}

\theoremstyle{remark}
\newtheorem{remark}[theorem]{Remark}

\newcommand{\R}{\mathbb{R}}

\newcommand{\E}{\mathbb{E}}
\newcommand{\var}{\mathrm{Var}}

\newcommand{\iid}{\stackrel{\mathrm{i.i.d.}}{\sim}}

\newcommand{\eqdist}{\stackrel{d}{=}}

\newcommand{\ind}{\mathbbm{1}}

\newcommand{\Cov}{\mathrm{Cov}}
\newcommand{\Corr}{\mathrm{Corr}}
\newcommand{\Var}{\mathrm{Var}}
\newcommand{\tr}{\mathrm{tr}}

\newcommand{\diag}{\mathrm{diag}}

\newcommand{\argmin}{\mathrm{argmin}}

\newcommand{\GEE}{\mathrm{GEE}}

\newcommand{\QML}{\mathrm{EQML}}

\newcommand{\SL}{\mathrm{SL}}

\newcommand{\GLM}{\mathcal{P}_\mathrm{GLM}}
\newcommand{\QGLM}{\mathcal{P}_\mathrm{QGLM}}
\newcommand{\varmodel}{\mathcal{P}_v}
\newcommand{\meanmodel}{\mathcal{P}_{\mathrm{\mu}}}
\newcommand{\estimators}{\mathcal{E}_{\textrm{QML}}}
\newcommand{\Hom}{\mathcal{P}_\mathrm{Hom}}
\newcommand{\Het}{\mathcal{P}_\mathrm{Het}}
\newcommand{\VVar}{\mathcal{P}_\mathrm{Var}}

\newcommand{\vertiii}[1]{{\left\vert\kern-0.25ex\left\vert\kern-0.25ex\left\vert #1 
    \right\vert\kern-0.25ex\right\vert\kern-0.25ex\right\vert}}
    
\newcommand\independent{\protect\mathpalette{\protect\independenT}{\perp}}
    \def\independenT#1#2{\mathrel{\rlap{$#1#2$}\mkern2mu{#1#2}}}
\newcommand\given{\,|\,}

\newcommand{\RN}[1]{
  \textup{\uppercase\expandafter{\romannumeral#1}}
}

\makeatletter
\newcommand{\mylabel}[2]{#2\def\@currentlabel{#2}\label{#1}}
\makeatother

\title{Sandwich regression for accurate and robust estimation in generalized linear multilevel and longitudinal models}

\date{December 2024}
\usepackage{authblk}%
\author[1]{Elliot H.\ Young
\thanks{ey244@cam.ac.uk}}
\author[1]{Rajen D.\ Shah
\thanks{r.shah@statslab.cam.ac.uk}}
\affil{Statistical Laboratory, University of Cambridge, UK}

\begin{document}

\maketitle

\begin{abstract}
Generalized linear models are a popular tool in applied statistics, with their maximum likelihood estimators enjoying asymptotic Gaussianity and efficiency. 
As all models are wrong, it is desirable to understand these estimators' behaviours under model misspecification. We study semiparametric multilevel generalized linear models, where only the conditional mean of the response is taken to follow a specific parametric form. Pre-existing estimators from mixed effects models and generalized estimating equations require specificaiton of a conditional covariance, which when misspecified can result in inefficient estimates of fixed effects parameters. 
It is nevertheless often computationally attractive to consider a restricted, finite dimensional class of estimators, as these models naturally imply. 
We introduce sandwich regression, that selects the estimator of minimal variance within a parametric class of estimators over all distributions in the full semiparametric model.
We demonstrate numerically on simulated and real data the attractive improvements our sandwich regression approach enjoys over classical mixed effects models and generalized estimating equations.
\end{abstract}

\section{Introduction}
All models are wrong, but some are useful. When a large quantity of high-quality data is available, it is natural to be minimal with regards to the model assumptions being imposed when performing statistical analyses, with assumption-lean estimators typically requiring data-driven estimation of nonparametric nuisance functions. 
When dealing with smaller datasets, these nonparametric estimations become impractical, and instead often lead one to impose certain structural assumptions on the data generating process. 
The generalized linear model (glm)~\citep{glms} has proved a popular structure to impose, 
where the distribution of an outcome $Y\in\R$ conditional on covariates $X\in\R^{1\times p}$ lies in a finite dimensional parametric class. Together with a random design on the covariates, this defines a class of distributions $\GLM$ on the pair $(Y,X)$. 
The maximum likelihood estimator corresponding to $\GLM$ enjoys consistency, asymptotic Gaussianity and efficiency within this model. 
This estimator remains efficient within a broader class of models~\citep{qml} provided the first two conditional moments implied by the glm hold i.e.~the model 
\begin{equation*}
    \QGLM := \Big\{P: 
    g\big(\E_P[Y\given X]\big) = X\beta
    ,\;
    \Var_P(Y\given X) = v(g^{-1}(X\beta))
    \Big\},
\end{equation*}
where $\beta\in\R^p$ is a fixed effect parameter of interest, $g$ is a strictly increasing link function, and $v$ is the variance function of the glm.   
This model is termed the quasi generalized linear model~\citep{qmle}, with corresponding estimator the quasi maximum likelihood (QML) estimator, and extends immediately to allow the variance function $v$ to be an arbitrary function of the mean. The model $\QGLM$ therefore provides extra flexibility and robustness over $\GLM$, but arguably is still restrictive in terms of the form of the conditional variance $\Var_P(Y\given X)$ being imposed. 
This restriction can be alleviated~\citep{glmbook, eqml} to allow the variance function to take an arbitrary function of the covariates
\begin{equation*}
    \varmodel := \Big\{P:
    g\big(\E_P[Y\given X]\big) = X\beta
    ,\;
    \Var_P(Y\given X) = \nu_\gamma(X)
    ,\;
    \gamma\in\Gamma
    \Big\},
\end{equation*}
for some class of variance functions $\{\nu_\gamma:\gamma\in\Gamma\}$ where $\Gamma\subseteq\R^q$ for some $q\in\mathbb{N}$. This setting is sometimes termed pseudo maximum likelihood estimation~\citep{gourieroux, ziegler}. 

Given i.i.d.~realisations $(Y_i,X_i)\in\R\times\R^p$ from a distribution (not necessarily in $\varmodel$), and for each $\gamma\in\Gamma$ defining a law in $\varmodel$, a QML estimator $\tilde{\beta}(\gamma)$ for $\beta$ can be constructed, thus implying a class $\estimators$ of estimators parametrised by the dispersion parameter $\gamma\in\Gamma$;
\begin{multline*}
    \estimators = 
    \bigg\{ 
    \tilde{\beta}(\gamma)\in\R^p :\;
    \sum_i\psi_{\text{QML}}(Y_i,X_i;\tilde{\beta}(\gamma))=0
    ,\\
    \psi_{\text{QML}}(Y_i,X_i,\beta) = \frac{(Y_i-g^{-1}(X_i\beta))}{g'(g^{-1}(X_i\beta))\,\nu_\gamma(X_i)}X_i^T,
    \;
    \gamma\in\Gamma
    \bigg\}
    ,
\end{multline*}
When the `true distribution' lies within the class $\varmodel$, then provided the `true $\gamma$' - that is the $\gamma\in\Gamma$ that correctly specifies the variance function - can be estimated consistently, the resulting QML estimator for $\beta$ is also efficient. In the semiparametric literature this property is termed \emph{local efficiency}~\citep{vanderlaan-robins, tsiatis, zhang, targeted-learning}.
In practice two prominent methods are commonly employed for dispersion estimation. 
The first is extended quasi maximum likelihood estimation (EQML)~\citep{eqml, glmbook}, that fits the pseudo-likelihood obtained if $Y\given X$ were normally distributed with variance $\nu_{\gamma}(X)$. 
Note that in practice the terms EQML and QML are often used synonymously; for avoidance of doubt here we use EQML to refer to dispersion parameter estimation and QML to refer to specifically $\beta$ estimation given a fixed dispersion $\gamma\in\Gamma$. 
The second strategy is common in the literature for generalized estimating equations (GEE), which aims to estimate the dispersion parameter via a least squares regression of the squared Pearson residuals onto the hypothesised (scaled) variance model. 
When the true distribution lies in $\varmodel$ both these strategies consistently estimate the optimal dispersion, resulting in locally efficient estimation. 
However, local efficiency says nothing with regards to the variance of the estimator when the true distribution lies outside this model $\varmodel$. We will study the broader model that only assumes the first conditional moment implied by the glm
\begin{equation*}
    \meanmodel := \Big\{P: g\big(\E_P[Y\given X]\big) = X\beta \Big\}.
\end{equation*}
Within $\meanmodel$ every estimator in $\estimators$ remains consistent and asymptotically Gaussian, albeit with differing variances. We will see that locally efficient estimators with respect to the class $\varmodel$ need not even be the estimator in $\estimators$ of minimal variance when the true distribution lies within $\meanmodel\backslash\varmodel$ i.e.~when the conditional variance model is misspecified, and in fact this suboptimality can be arbitrarily large (see Proposition~\ref{prop:divratio}). 
We begin by exploring a simple motivating example of the above in the clustered linear model (see Section~\ref{sec:motivating-eg}). First we introduce a few details specific to the linear model.

\begin{figure}[ht]
    \centering
    \includegraphics[width=0.9\textwidth]{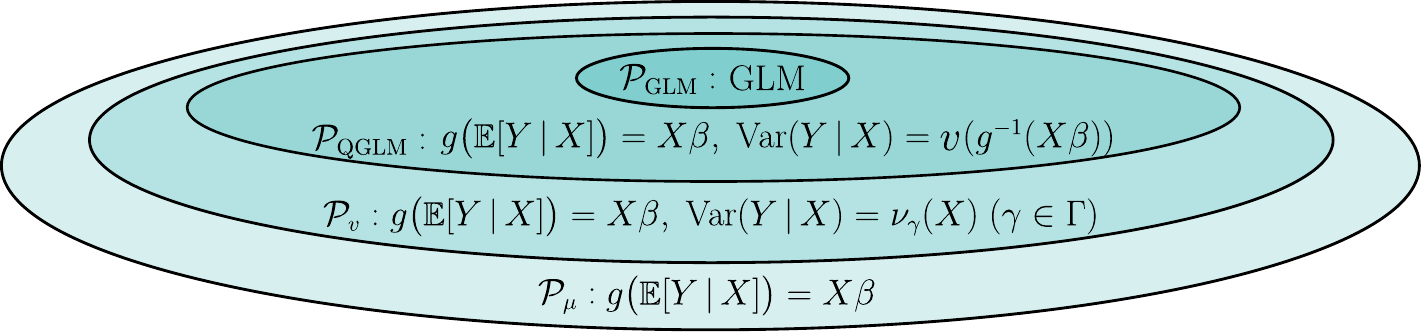}
    \caption{Schematic to describe the hierarchy of models $\GLM\subset\QGLM\subset\varmodel\subset\meanmodel$ we will consider. In our setup we consider a class of estimators $\estimators$ given by the quasi maximum likelihood $\beta$ estimators obtained as laws span over $\varmodel$ (i.e.~variance functions spanning some potentially misspecified parametric class $\{\nu_\gamma:\gamma\in\Gamma\}$). We will be selecting from this class of estimators that of minimal variance under only the assumption that the true law lies in $\meanmodel$.}
    \label{fig:hierarchy}
\end{figure}

\subsection{Grouped data}\label{sec:grouped-data}

The clustered quasi-Gaussian linear model is an example of a quasi generalized linear model for grouped data, asserting each cluster, consisting of a multivariate $Y$ and matrix of covariates $X$, satisfies
\begin{equation}\label{eq:quasi-LM}
    \E_P[Y\given X]=X\beta,
    \qquad
    \Cov_P(Y\given X) = \Sigma_\gamma(X)  \;(\gamma\in\Gamma),
    \end{equation}
    for some variance model defined in terms of a dispersion parameter $\gamma\in\Gamma$. 
    Given i.i.d.~data $(Y_i,X_i)\eqdist (Y,X)$ the class of QML estimators in the quasi-Gaussian linear model takes the form of the weighted least squares estimators
    \begin{equation}\label{eq:lin-estimators}
        \estimators = \bigg\{\tilde{\beta}(\gamma)=\Big(\sum_i X_i^T\Sigma_\gamma(X_i)^{-1}X_i\Big)^{-1}\Big(\sum_i X_i^T\Sigma_\gamma(X_i)^{-1}Y_i\Big)
        :\,\gamma\in\Gamma
        \bigg\}.
    \end{equation}
    A popular dispersion structure imposed for clustered data is the linear mixed effects model, which asserts
    \begin{equation*}
    Y\given X  \sim  N\big( X\beta \,,\; X^{(\text{sub})}\mathcal{V}X^{(\text{sub})T} + \sigma^2 I \big),
\end{equation*}
for some subcollection $X^{\text{(sub)}}$ of the full covariates $X$, positive definite matrix $\mathcal{V}$ and $\sigma>0$, which motivates the class of weighted least squares estimators~\eqref{eq:lin-estimators} with $$\Sigma_\gamma(X)=X^{\text{(sub)}}\mathcal{V}X^{\text{(sub)}T}+\sigma^2 I, \quad \gamma=(\mathcal{V},\sigma).$$ 
The quadratic conditional covariance introduced here is motivated through the imposition of fictitious `random effects', whose aim is to motivate an estimator of reduced variance to the inefficient OLS estimator. 
For longitudinal datasets the ARMA model is often used to imply a practical dispersion structure. 
Given a working model $\Sigma_\gamma$, to estimate the dispersion parameter $\gamma$ the EQML and GEE losses for the grouped linear model~\eqref{eq:quasi-LM} correspond to 
\begin{align}
    \gamma_{P,\QML} &:= \underset{\gamma\in\Gamma}{\mathrm{argmin}}\,\E_P\Big[\log\Sigma_\gamma(X)+\varepsilon^T\Sigma_\gamma(X)^{-1}\varepsilon\Big] \notag
    \\
    &\,= \underset{\gamma\in\Gamma}{\mathrm{argmin}}\,\E_P\Big[\log\Sigma_\gamma(X)+\tr\big\{\Sigma_\gamma(X)^{-1}\Cov_P(Y\given X)\big\}\Big], \notag
    \\
    \gamma_{P,\GEE} &:= \underset{\gamma\in\Gamma}{\mathrm{argmin}}\,\E_P\Big[\big\|\Sigma_\gamma(X)-\varepsilon\varepsilon^T\|_F^2\Big] \notag
    \\
    &\,=\underset{\gamma\in\Gamma}{\mathrm{argmin}}\,\E_P\Big[\big\|\Sigma_\gamma(X)-\Cov_P(Y\given X)\|_F^2\Big] 
    , \label{eq:EQML-GEE-pop}
\end{align}
at the population level, where $\|\cdot\|_F$ denotes the Frobenius norm. In the following example we explore the potential suboptimality of imposing the EQML loss to estimate the dispersion parameters in $\estimators$~\eqref{eq:lin-estimators}, with optimality given in terms of the variance of the $\beta$ estimator. For similar examples also studying the GEE loss function see~\citet{sandboost}.

\subsection{Failure of (extended) quasi-maximum likelihood and generalized estimating equations}

\subsubsection{Longitudinal data example}\label{sec:motivating-eg}

We consider an example in the clustered linear model~\eqref{eq:quasi-LM}, where the conditional covariance is misspecified. 
Suppose we have data consisting of i.i.d.~groups each of size $50$ and following the data generating mechanism
\begin{gather*}
    Y\given X \sim N_{50}\big(X\beta,\,\Omega\big),
    \qquad
    X\sim N_{50}\big({\bf 0},\,0.9\,{\bf 1}{\bf 1}^T+0.1\,I_{50}\big)
    \\
    \Omega_{jk} = \begin{cases}
        1 &\text{if } j=k,
        \\
        e^{-\frac{1}{2}|j-k|^{\frac{1}{4}}} + \frac{1}{4}\cos(|j-k|)e^{-\frac{1}{20}|j-k|} \;&\text{if } j\neq k.
    \end{cases}
\end{gather*}
where $\beta=1$ is the estimand of interest. Note that by construction the model is homoscedastic, with correlation components bounded in $[0.24,0.74]$, and with $\Omega$ of condition number approximately that of an $\text{AR}(1)$ covariance matrix of the same group size and autoregressive parameter $0.9$. 
We consider the classes of weighted least squares estimators~\eqref{eq:lin-estimators} with weights $\Sigma_\gamma$ taking the form of the covariance matrices of nested $\text{ARMA}(p,q)$ processes with $(p,q)\in\{1,2\}\times\{0,1,2\}$. To compare these dispersion models, the Akaike information (AIC) is often used as a quality of fit metric, with variations including: Bayesian information (BIC); corrected AIC (AICc) for longitudinal analysis~\citep{Brockwell}; and the quasi information (QIC) when solving GEEs~\citep{QIC, hardin}. Both AIC and BIC are commonly used to compare mixed effects model structures~\citep{membook}. 

Table~\ref{tab:intro} presents the AIC of each ARMA model fit using (extended) quasi maximum likelihood with each of the ARMA working covariance models, in addition to the mean squared error of their $\beta$ estimators. The colour shading reflects the preferred ordering of the models in terms of AIC and $\beta$ mean squared error. 
We see that out of the six ARMA models, the $\text{AR}(1)$ model appears worst in terms of AIC, yet it is the best model in terms of minimal mean squared $\beta$ estimation error. Similarly, the `best model' in terms of AIC is in fact the second worst model in terms of mean squared $\beta$ estimation error. As such, using quasi-maximum likelihood estimation for estimating dispersion parameters alongside AIC for model selection can select not only a suboptimal dispersion parameter within a single model but also a worse dispersion model altogether. 

\begin{table}[ht]
	\begin{center}
        		\begin{tabular}{cc|cccc}
         & & \multicolumn{3}{c}{$q$}
         \\
	& $N^{-1} \text{AIC}$ & 0 & 1 & 2
        \\
        \hline
  \multirow{2}{*}{$p$} & 1 & \cellcolor{rank06}2.07719 & \cellcolor{rank04}2.04449 & \cellcolor{rank02}2.03271
  \\
  & 2 & \cellcolor{rank05}2.04497 & \cellcolor{rank03}2.04310 & \cellcolor{rank01}2.03269
        \end{tabular}
\quad
\begin{tabular}{c|cccc}
          & \multicolumn{3}{c}{$q$}
         \\
	$N \cdot \text{MSE}(\hat{\beta})$ & 0 & 1 & 2
        \\
        \hline
  1 & \cellcolor{rank01}3.51 & \cellcolor{rank02}3.59 & \cellcolor{rank06}4.33
  \\
  2 & \cellcolor{rank04}3.81 & \cellcolor{rank03}3.70 & \cellcolor{rank05}3.95
        \end{tabular}
	\caption{The AIC and MSE for estimators of $\beta$ in the longitudinal data example of Section~\ref{sec:motivating-eg}. The model fit is a univariate linear model with errors following an $\text{ARMA}(p,q)$ model for $(p,q)\in\{1,2\}\times\{0,1,2\}$. 
 Models are fit using the Gaussian quasi-MLE objective (EQML) via the \texttt{nlme} package \citep{nlme-code}. 
 The empirical sample size used is $N=10^6$, at which size the AIC ordering of the models are identical to those for BIC, AICc and QIC.}\label{tab:intro}
	\end{center}
\end{table}

\subsubsection{Heteroscedastic linear model}

Whilst when the true law lies within the variance model $\varmodel$ the optimal dispersion parameter (that is the dispersion parameter that corresponds to the true conditional variance) is estimated consistently by a range of loss functions (including the EQML and GEE loss), when the true law lies within $\meanmodel\backslash\varmodel$ it is less clear precisely what the dispersion parameter being estimated by these objective functions converge to. As we see in the motivating example of Section~\ref{sec:motivating-eg}, when the quasi-maximum likelihood EQML loss is used to estimate dispersion, the resulting $\beta$ estimator need not coincide with the $\beta$ estimator of minimal variance amongst the class of estimators $\estimators$ as in~\eqref{eq:lin-estimators}. In this section we explore this unfavourable behaviour in a theoretical framework.

In the remainder of this section we focus our attention on the ungrouped, univariate, heteroscedastic Gaussian linear model $\Het$, with the nested submodels
\begin{align*}
    \Hom &:= \Big\{P: Y\given X\sim N\big(X\beta,\,\sigma^2\big),\,\sigma>0\Big\}
    \\
    \VVar &:= \Big\{P: Y\given X\sim N\big(X\beta,\,\nu_\gamma(X)\big),\,\gamma\in\Gamma\Big\}
    \\
    \Het &:= \Big\{P: Y\given X\sim N\big(X\beta,\,\nu_\gamma(X)\big), \, \E[\nu_\gamma(X)]<\infty, \, \nu_\gamma(X)>0 \text{ a.s.}\Big\},
\end{align*}
for a class of variance functions $\{\nu_\gamma:\gamma\in\Gamma\}$, and with random covariate design. In relation to the schematic of Figure~\ref{fig:hierarchy} these models lie in $\Hom\subset\GLM$, $\VVar\subset\varmodel$ and $\Het\subset\meanmodel$; we introduce these submodels to strengthen our theoretical results (Proposition~\ref{prop:divratio} to follow). In this ungrouped setting the class of weighted least squares estimators $\estimators$ takes the form~\eqref{eq:lin-estimators} with $\Sigma_\gamma=\nu_\gamma$, with $\nu_\gamma$ a scalar-valued variance function. 
In particular, we show there exists distributions in the stricter Gaussian model $\Het$ 
where optimising the EQML or GEE loss for dispersion estimation in the class of  weighted least squares estimators~\eqref{eq:lin-estimators} performs suboptimally in terms of the asymptotic variance of the $\beta$ estimator 
\begin{equation}\label{eq:V(gamma)}
    V_P(\gamma) := {\E_P\bigg[\frac{X^2}{\nu_\gamma(X)}\bigg]}^{-2} 
    {\E_P\bigg[\frac{\Var_P(Y\given X)\,X^2}{\nu_\gamma(X)^2}\bigg]} .
\end{equation} 
We call this function the \emph{sandwich loss}, which we later generalise to multilevel generalized linear models (see Section~\ref{sec:SL}), with this terminology motivated from the sandwich estimator of~\citet{huber}.

One might intuitively expect 
if the `working variance' model $\{\nu_\gamma:\,\gamma\in\Gamma\}$ is misspecified but close to the true model then near-efficient estimation could follow; this is perhaps the general thinking behind a practitioner's use of a specific working variance model. 
However we show in the following proposition this is not necessarily the case. 
Recall that the Kullback--Leibler divergence between distributions  $P_1$ and $P_2$ both absolutely continuous with respect to Lebesgue measure is given by
\begin{equation*}
    \mathrm{KL}(P_1\,\|\,P_2) := \int \log\bigg(\frac{dP_1}{dP_2}\bigg)\,dP_1.
\end{equation*}

\begin{prop}\label{prop:divratio}
    Consider the semiparametric, ungrouped, univariate, linear model $\Het$. Also consider the asymptotic (population-level)  variances~\eqref{eq:V(gamma)} of the quasi-maximum likelihood (weighted least squares) estimators corresponding to the working variance model
    \begin{equation}\label{eq:working-var}
        \nu_\gamma
        \in
        \big\{
        x \mapsto \gamma_1\,\ind_{[0,\infty)}(x) + \gamma_2\,\ind_{(-\infty,0)}(x)
        \,:\, \gamma=(\gamma_1,\gamma_2)\in\R_+^2
        \big\},
    \end{equation}
    for the EQML and GEE methods given by~\eqref{eq:EQML-GEE-pop}. 
    Then for any $\tau,\sigma>0$, arbitrarily small $\alpha>0$ and arbitrarily large $\eta\geq1$, there exists a distribution $P\in\Het$ with 
    \begin{gather}
        \inf_{\tilde{P}\in\Hom}\mathrm{KL}(\tilde{P}\,\|\,P)\leq\alpha
        \quad\;
        \Var_P(X)=\tau^2,
        \quad\;
        \E_P[\Var_P(Y\given X)] = \sigma^2, \label{eq:conditions}
        \\
        \frac{V_P(\gamma_{P,\QML}) \wedge V_P(\gamma_{P,\GEE})}{\inf_{\gamma\in\Gamma}V_P(\gamma)} \geq\eta. 
        \label{eq:div-ratio}
    \end{gather}
\end{prop}

It is perhaps worrying that taking Proposition~\ref{prop:divratio} with $\eta$ arbitrarily large that the EQML and GEE dispersion estimators can result in arbitrarily suboptimal $\beta$ estimation within the class of estimators $\estimators$. Further, by taking $\alpha>0$ arbitrarily small we see that this can occur for distributions arbitrarily close to the working variance model imposed, and even arbitrarily close to the homoscedastic model. It follows that the optimality of dispersion parameter estimation when the working (co)variance model is misspecified crucially hinges on the objective function used, with EQML and GEE based optimisations potentially performing rather poorly. 
These issues however can be overcome by simply proposing~\eqref{eq:V(gamma)} as the loss function for dispersion estimation, thereby targetting the variance of the $\beta$ estimator directly. 
This framework allows for a practitioner to achieve the best of both worlds; sufficiently restricting the class of estimators to be both computationally practical and avoid overfitting, whilst also maintaining optimal accuracy within the class of estimators proposed against an expansive class of statistical models.

\subsection{Organisation of the remainder of the paper}

From hereon we will focus on the setting of clustered data in the multilevel generalized linear model, due to variance gains tending to be more pronounced in practice in this setting. Our methodology however is equally applicable to unclustered generalized linear models. The remainder of this paper is organised as follows. In Section~\ref{sec:lit-rev} we review some relevant literature on multivariate generalized linear models. We outline the semiparametric multilevel generalized linear model in Section~\ref{sec:model}, and introduce the parametric class of fixed effects estimators from which our goal is to select amongst this class the estimator of minimal variance within the full semiparametric model. 
In Section~\ref{sec:sand-reg} we formally define the sandwich loss for this model, and introduce an empirical approximate Jackknife estimator for it in Section~\ref{sec:SL}. 
An estimator for the standard errors of the resulting fixed effects estimators 
is proposed in Section~\ref{sec:variance}, with applications to model selection discussed in Section~\ref{sec:model-selection}. Numerical experiments are explored in Section~\ref{sec:numericals}, both on simulated data and two real-world datasets.

\section{The semiparametric multilevel generalized linear model}\label{sec:model}

Suppose we observe $I$ independent data clusters, with $i$th cluster $(Y_i,X_i)\in\R^{n_i}\times\R^{n_i\times p}$ of group size $n_i$ satisfying the semiparametric multilevel generalized linear model 
\begin{equation}\label{eq:model}
    g\big(\E[Y_i\given X_i]\big) = X_i\beta,
\end{equation}
for some population averaged estimand $\beta\in\R^p$ and known, strictly increasing link function $g$ that we take with an abuse of notation to apply component-wise to a vector in $\R^{n_i}$, and where the covariate design may be interpreted as either random or fixed. Note specifically we make no assumptions on the form of the conditional covariance $\Cov(Y_i\given X_i)$. 
A semiparametric efficient estimator for $\beta$ in this model solves the estimating equation
\begin{gather*}
    \sum_{i=1}^I\psi_{\text{eff}}(Y_i,X_i;\beta)=0
    ,
    \qquad
    \psi_{\text{eff}}(Y,X;\beta) := D(X;\beta)^T\,\Cov(Y\given X)^{-1}\,\big(Y-g^{-1}(X\beta)\big),
    \\
    \big(D(X;\beta)\big)_{jk} := \Big(\frac{1}{g'(g^{-1}(X\beta))}\circ X\Big)_{jk} = \frac{X_{jk}}{g'(g^{-1}(X_{j\boldsymbol{\cdot}}\beta))}
    .
\end{gather*}
where $\psi_{\text{eff}}$ is known as the (scaled) semiparametric efficient influence function; see for example~\citet{tsiatis} for a derivation. This estimator is that of minimal variance over all regular, asymptotically linear estimators for $\beta$ over all distributions in~\eqref{eq:model}. 
To circumvent the need to estimate the non-parametric matrix-valued nuisance function $\Cov(Y\given X)$ 
we consider the class $\estimators$ of quasi-maximum likelihood estimators given by the estimating equations 
\begin{multline}\label{eq:estimators}
    \estimators := 
    \bigg\{
    \tilde{\beta}(\gamma)\in\R^p :\,
    \sum_{i=1}^I \psi(Y_i,X_i;\tilde{\beta}(\gamma),\gamma) = 0
    ,\;
    \psi(Y_i,X_i;\beta,\gamma) = D_i(\beta)W_i(\beta,\gamma)\big(Y_i-\mu_i(\beta)\big)
    ,\\
    D_i(\beta) := \frac{\partial\mu_i(\beta)}{\partial\beta^T}
    ,\;
    \mu_i(\beta) := g^{-1}(X_i\beta)
    ,\;
    W_i(\beta,\gamma) := \Sigma_i(\gamma)^{-1} = A_i(\beta)^{-\frac{1}{2}}
    \mathrm{P}_i(\gamma)^{-1}A_i(\beta)^{-\frac{1}{2}}
    ,\\
    A_i(\beta) := \diag\big(v(\mu_i(\beta))\big)
    ,\;
    \gamma\in\Gamma
    \bigg\},
\end{multline}
where $\mathrm{P}$ is a pseudo-correlation matrix (that may also model variances in misspecified scenarios), $\Sigma$ a full working variance model, $v$ the generalized linear model's variance function, 
and $\Gamma$ a Euclidean space spanning the possible dispersion parameters $\gamma\in\Gamma$. We explore a few examples of models that imply dispersion structures on $\mathrm{P}_i(\gamma)$ and $\Sigma_i(\gamma)$. From hereon we denote dependence on $X_i$ in quantities such as $\mathrm{P}_i(\gamma)$, $\Sigma_i(\gamma)$, $D_i(\beta)$, $\mu_i(\beta)$, $W_i(\beta,\gamma)$ through the $i$ subscript. 
\begin{example}[Classes of Mixed Effects Models]
    Mixed effects models~\citep{glmmbook} have proved popular in practice. Linear mixed effects models, as well as Tweedie generalized mixed effects models~\citep{lee-nelder, lee-nelder-book} with multiplicative random effects and a log-link function, have the popular robustness properties that the `subject specific' and `population averaged' quantities coincide. The Tweedie family of distributions allow for flexible modeling of positive continuous (e.g.~zero-inflated) data as well as count data, with hierarchical and longitudinal extensions proposed by \citet{ma} and \citet{holst} respectively.
\end{example}
\begin{example}[Multivariate Dispersion Families]
    \citet{medf, medf2} construct a class of multivariate dispersion model for generalized linear models with multivariate responses. These provide a general and flexible class of models that include the elliptically contoured distributions.
\end{example}
\begin{example}[ARMA, GARCH, and ARMA-GARCH models]
    For longitudinal data ARMA \citep{box} and GARCH \citep{garch} model structures can be used to model conditional correlation and variance respectively, with ARMA-GARCH models combining both. Together, these imply a parametric structure on $\Cov(Y_i\given X_i)$ 
    which is commonly estimated by (extended quasi) maximum likelihood. See \citet{box} and \citet{Brockwell} for a review.
\end{example}

\subsection{Some related literature in multilevel models}\label{sec:lit-rev}

Given a dispersion model, such as those outlined in the previous subsection, two estimation strategies have dominated both theory and practice. Both follow the general paradigm of a weighted least squares regression of squared residuals to a given parametric covariance structure, namely solving the estimating equation 
\begin{equation}\label{eq:GEE.MLE.disp}
    \sum_{i=1}^I \tr\big( \mathcal{W}_i (\hat{r}_i\hat{r}_i^T - \Sigma_i(\gamma)) \big) = 0,
\end{equation}
for the dispersion parameter $\gamma$, where $\hat{r}_i = Y_i-g^{-1}(X_i\hat{\beta})$ denotes the residuals of $I$ clustered multivariate datapoints, $\Sigma_i(\gamma)$ is some dispersion matrix parametrisation and $\mathcal{W}_i$ is some weighting matrix that will vary depending on the method considered.  
\cite{eqml, severini1} propose the extended quasi maximum likelihood (EQML) estimating equation, that corresponds to maximising the likelihood were the residuals to follow a mean zero multivariate normal distribution with covariance $\Sigma_i(\gamma)$, subsequently adopted by \citet{ma, holst, jorg-knudsen} among others. 
The second method is GEE1~\citep{gee1} (GEE), which determines the dispersion parameters by minimising a weighted least squares on the cross-products of the Pearson residuals \citep{zeger-liang, hardin, ziegler}. In theory, any residual could fall within this GEE framework, with Pearson residuals often employed in practice. In summary, these methods solve the estimating equation~\eqref{eq:GEE.MLE.disp} with
\begin{equation}\label{eq:EQML-GEE}
    \mathcal{W}_i = \begin{cases}
        \Sigma_i(\gamma)^{-1}\frac{\partial\Sigma_i(\gamma)}{\partial\gamma}\Sigma_i(\gamma)^{-1} &\text{for EQML \citep{severini1}}
        \\
        \diag\big(v(\mu_i(\hat{\beta}))\big)^{-1} \frac{\partial\Sigma_i(\gamma)}{\partial\gamma}
        \diag\big(v(\mu_i(\hat{\beta}))\big)^{-1} &\text{for GEE \citep{zeger-liang}}
    \end{cases},
\end{equation}
where $v$ denotes the glm's variance function and $\hat{\beta}$ is a pilot estimator for $\beta$.  
In finite samples both of these objectives have been adapted by adding a perturbation to the estimating equations such that the estimating equation~\eqref{eq:GEE.MLE.disp} remains unbiased \citep{mccullagh-tib, severini2, jorg-knudsen}.

It is often noted that the homoscedasticity assumption, for which variance estimators of fixed effects estimators in generalized linear models often rely on, is often violated in practice. The sandwich estimator of the variance~\citep{huber, white, gourieroux2, royall, zeger-liang} is often employed to allow for asymptotically valid estimates of standard errors in heteroscedastic settings. Finite sample adaptations have been proposed in ordinary least squares regression, including so called `heteroscedasticity consistent' estimators of the variance (HC1--5) \citep{mackinnon-white, HC4-5}, which in recent years has been extended to clustered data~\citep{mackinnon2023fast, grouped-leverage}.

\section{Sandwich regression}\label{sec:sand-reg}

\subsection{The finite sample empirical sandwich loss}\label{sec:SL}

Consider the target of interest being a linear combination $c^T\beta$ of the fixed effects, for some deterministic vector $c\in\R^p$, with an attributed notion of optimality for an estimator $\hat{\beta}$ being minimality of $\Var(c^T\hat{\beta}) = c^T\Var(\hat{\beta})c$ in the class of QML estimators $\estimators$~\eqref{eq:estimators}. 
Recall sandwich regression estimates the optimal dispersion $\gamma$ in $\estimators$ by minimising an estimate of the varaince of the scalar target of interest, in this case $\var(c^T\tilde{\beta}(\gamma))$, which takes the form
\begin{align}\label{eq:V(gamma)-0}
    V(\gamma) &:= \textrm{M}(\gamma)^{-1}\textrm{S}(\gamma)\textrm{M}(\gamma)^{-1},
    \\
    \textrm{M}(\gamma) &:= \frac{1}{I}\sum_{i=1}^I\E\Big[D_i(\tilde{\beta}(\gamma))^TW_i(\tilde{\beta}(\gamma),\gamma) D_i(\tilde{\beta}(\gamma))\Big] \notag
    \\
    \textrm{S}(\gamma) &:= \frac{1}{I}\sum_{i=1}^I\E\Big[D_i(\tilde{\beta}(\gamma))^TW_i(\tilde{\beta}(\gamma),\gamma)\Cov(Y_i\given X_i)W_i(\tilde{\beta}(\gamma),\gamma)D_i(\tilde{\beta}(\gamma))\Big]
    . \notag
\end{align} 
\citet{sandboost} introduce an asymptotically unbiased estimator for this sandwich loss, which performs well in large samples but in finite samples is downwards biased. We therefore introduce the finite sample empirical sandwich loss, given by
\begin{equation}\label{eq:small-SL}
    L_{\SL}(\gamma) := c^T\tilde{V}(\gamma)c,
    \qquad
    \tilde{V}(\gamma) := \frac{1}{I}\sum_{i=1}^I \big(\tilde{\beta}_{(-i)}(\gamma) -\tilde{\beta}(\gamma)\big)\big(\tilde{\beta}_{(-i)}(\gamma) -\tilde{\beta}(\gamma)\big)^T,
\end{equation}
where $\tilde{\beta}_{(-i)}(\gamma)$ solves the analogous estimating equation of~\eqref{eq:estimators} on the $i$th cluster leave-one-out sample. 
This estimator approximates the jackknife estimator of $\Var(\tilde{\beta}(\gamma))$ as in \citet{efron} for a fixed $\gamma$. 
In the linear ($g=\text{id}$), ungrouped ($n_i\equiv 1$), unweighted ($W_i$ constant) setting this estimator reduces to the `hetersocedastic consistent' (HC3) estimator of the OLS estimator's variance; see for example \citet{mackinnon-white, HC4-5} for an overview and \citet{mackinnon2023fast} for a number of relevant simulations. In recent years this HC3 estimator has been extended to the cluster regression setting \citep{grouped-leverage} for the OLS estimator. 

The sandwich loss~\eqref{eq:small-SL} can be calculated in $O(n)$ operations for the linear model, and can be approximated for the generalized linear model at the same computational order, as outlined in Algorithm~\ref{alg:SL}. The computational approximation methods employed here are analogous to those used when performing leave-one-out cross-validation in glms~\citep{hardin, aloocv} by minimising a locally quadratic approximation about the full-sample solution $\tilde{\beta}(\gamma)$ of the leave-one-out objective; 
\begin{equation}\label{eq:beta_min_i}
    \tilde{\beta}_{(-i)}(\gamma) \approx \tilde{\beta}(\gamma) - \bigg(\sum_{i'=1}^I D_{i'}^TW_{i'}D_{i'} - D_i^TW_iD_i\bigg)^{-1} \bigg(D_i^TW_i\big(Y_i-\mu_i(\tilde{\beta}(\gamma))\big)\bigg),
\end{equation}
where for notational convenience we take $D_i:=D_i(\tilde{\beta}(\gamma))$ and $W_i:=W_i(\tilde{\beta}(\gamma),\gamma)$. 
Note that this approximation is exact in the linear model.

{
\RestyleAlgo{ruled}
\begin{algorithm}[ht
]
\SetKwInOut{Notation}{Notation}
\KwIn{Generalized linear model with working dispersion structure as in~\eqref{eq:estimators};
dispersion parameter $\gamma\in\Gamma$ at which the sandwich loss $L_{\SL}(\gamma)$ is to be calculated;
vector $c\in\R^p$ dictating sandwich regression's target.
}

\Notation{On matrix and vector quantities that are dependent on the full-sample $\beta$ estimator $\tilde{\beta}(\gamma)$ and $\gamma$, we omit this dependence i.e.~we notate $D_i = D_i(\tilde{\beta}(\gamma))$, $W_i=W_i(\tilde{\beta}(\gamma),\gamma)$ etc. We also notate $W := \diag(W_1,\ldots,W_I)$, $D := (D_1, \ldots, D_I)$ and their leave-one-out analogies with subscript $(-i)$. We also allow $g^{-1}$ to act component-wise when applied to a vector i.e.~$(g^{-1}(\eta))_j = g^{-1}(\eta_j)$.}

Calculate:
\\
     \quad
     \textbullet\, $\tilde{\beta}(\gamma)$ as the solution of the estimating equation~\eqref{eq:estimators},
     \\
     \quad
     \textbullet\, $\big(D_i^TW_iD_i\big)_{i\in[I]}$ and $D^TWD = \sum_{i=1}^I D_i^TW_iD_i$,
\\
either analytically (in the linear model) or by Fisher-scoring (in the glm). 

\For{$i\in[I]$} {

    $R_i := Y_i - g^{-1}(X_i\tilde{\beta}(\gamma))$,

    $U_i := W_i^{-1} - D_i\big(D^TWD\big)^{-1}D_i^T$,

    $T_i := \big(D_{(-i)}^TW_{(-i)}D_{(-i)}\big)^{-1} = \big(D^TWD\big)^{-1} + \big(D^TWD\big)^{-1}\big(D_i^TU_i^{-1}D_i\big)\big(D^TWD\big)^{-1}$,

    $S_i := D_i^TW_iR_iR_i^TW_iD_i$
    
}

Calculate $\tilde{V}(\gamma) := \sum_{i=1}^I T_iS_iT_i$

Calculate 
$L_{\SL}(\gamma) := c^T\tilde{V}(\gamma)c$

\KwOut{Empirical sandwich loss $L_{\SL}(\gamma)$ evaluated at $\gamma$.}
\caption{Calculating the sandwich loss $L_{\SL}(\gamma)$}
\label{alg:SL}
\end{algorithm}
}

\subsection{Estimating standard errors for fixed effects}\label{sec:variance}

 Often confidence interval construction and test statistics in mixed effects models and GEEs rely on the asymptotic assumption that the dispersion parameter estimator $\hat{\gamma}$ converges in probability to some deterministic vector $\gamma$ at $\sqrt{n}$-rate. Whilst this is reasonable in an asymptotic regime, for small or even moderate samples not accounting for finite sample corrections can result in undercoverage of confidence intervals. These issues for mixed effects models and GEEs at finite samples also pervade sandwich regression. 
 We therefore introduce an estimator of the (co)variance of sandwich regression $\beta$ estimators that takes into account the variation in dispersion estimation. This takes the form of a Jackknife estimator as in \citet{efron}. Note ordinary heteroscedasticity-robust (HC3) estimators of the variance \citep{mackinnon-white,mackinnon2023fast,grouped-leverage,HC4-5} study a Jackknife estimator of the variance for the OLS estimator; we consider a generalisation that additionally allows for dispersion estimation, as well as extends to the generalized linear model.

The estimator for the variance of the sandwich regression $\hat{\beta}$ estimator is given in Algorithm~\ref{alg:V-hat}. We highlight a few key components of this algorithm. First we introduce some notation. Denoting the sandwich loss~\eqref{eq:small-SL} as $L_{\SL}$ with the leave-($i$)-out version $L_{\SL(-i)}$ as~\eqref{eq:small-SL} omitting the $i$th observation, each with minimisers $\hat{\gamma}$ and $\hat{\gamma}_{(-i)}$ respectively. Also recall the functions $\tilde{\beta}, \tilde{\beta}_{(-i)} : \Gamma\to\R^p$. Then we define the terms
\begin{gather*}
    \tilde{\gamma} = \underset{\gamma\in\Gamma}{\argmin} \,L_{\SL}(\gamma),
    \qquad
    \tilde{\gamma}_{(-i)} = \underset{\gamma\in\Gamma}{\argmin} \,L_{\SL(-i)}(\gamma),
    \\
    \hat{\beta}_{(-i)} = \tilde{\beta}_{(-i)}(\hat{\gamma}_{(-i)}),
    \qquad
    \check{\beta}_{(-i)} = \tilde{\beta}_{(-i)}(\hat{\gamma}),
    \qquad
    \hat{\beta} = \tilde{\beta}(\hat{\gamma}).
\end{gather*}
All these terms are used as intermediaries to calculate the Jackknife variance estimator $\hat{V}$, with $\hat{\gamma}_{(-i)}$, $\hat{\beta}_{(-i)}$ and $\check{\beta}_{(-i)}$ approximated. These are made computationally tractable by using Newton-type iterative steps of a locally quadratic approximation of their respective leave-one-out objectives about their full-sample minimisers. 
Algorithm~\ref{alg:V-hat} performs one iterative step to approximate these quantities. 
In substantially smaller samples, where both the above locally quadratic approximation may be unreasonable and the computational considerations at play are less of a burden, multiple iterative steps of the above optimisation can be performed, editing lines 5, 6 and/or 8 of Algorithm~\ref{alg:V-hat} accordingly.

{
\RestyleAlgo{ruled}
\begin{algorithm}[ht]
\SetKwInOut{Notation}{Notation}
\SetKwComment{Comment}{$\triangleright$\ }{}
\KwIn{Generalized linear model with working dispersion structure as in~\eqref{eq:estimators};
estimator for the dispersion parameters $\hat{\gamma}\in\argmin_{\gamma\in\Gamma}L_{\SL}(\gamma)$ (as calculated in Algorithm~\ref{alg:SL}) and corresponding $\beta$ estimator $\hat{\beta}=\tilde{\beta}(\hat{\gamma})$.
}

\Notation{On matrix and vector quantities dependent on $\hat{\beta}$ we omit its dependence i.e.~$D_i=D_i(\hat{\beta})$, $W_i=W_i(\hat{\beta},\hat{\gamma})$ etc. We also notate $W=\diag(W_1,\ldots,W_I)$ and $D=(D_1,\ldots,D_I)$, with the leave-one-out analogies given subscripts $(-i)$.}

Calculate $\hat{R} := Y - g^{-1}(X\hat{\beta}).$

\For{$i \in [I]$} {

$U_i := W_i^{-1} - D_i\big(D^TWD\big)^{-1}D_i^T$.

$T_i := \big(D_{(-i)}^TW_{(-i)}D_{(-i)}\big)^{-1} = \big(D^TWD\big)^{-1} + \big(D^TWD\big)^{-1}\big(D_i^TU_i^{-1}D_i\big)\big(D^TWD\big)^{-1}$.

$\check{\beta}_{(-i)} - \hat{\beta} \approx - T_i \big(D_i^TW_i\hat{R}_i\big) $
 
$\hat{\gamma}_{(-i)} - \hat{\gamma} \approx \Big(
        \frac{\partial^2L_{\SL (-i)}}{\partial\gamma\partial\gamma^T}\Big|_{\gamma=\hat{\gamma}}\Big)^{-1}\Big(\frac{\partial (L_{\SL (-i)}-L_{\SL})}{\partial\gamma}\Big|_{\gamma=\hat{\gamma}}\Big)$ 

$\Delta_{(-i)} := W_{(-i)}(\hat{\gamma}_{(-i)}) - W_{(-i)}(\hat{\gamma})$.

$\hat{\beta}_{(-i)} - \hat{\beta} \approx
     \big(I_p - T_i\big(D_{(-i)}^T\Delta_{(-i)}D_{(-i)}\big) \big)
     \big(\check{\beta}_{(-i)}-\hat{\beta}\big)
     + 
     T_i\big(D_{(-i)}^T\Delta_{(-i)}\hat{R}_{(-i)}\big),$

}

$\hat{V} = \frac{I-1}{I}\sum_{i=1}^I\big(\hat{\beta}_{(-i)}-\hat{\beta}\big)\big(\hat{\beta}_{(-i)}-\hat{\beta}\big)^T$.

\KwOut{An estimator $\hat{V}$ for $\Var(\hat{\beta})$.}
\caption{Constructing an estimator for $\Var(\hat{\beta})$.}
\label{alg:V-hat}
\end{algorithm}
}

\subsection{Dispersion model selection}\label{sec:model-selection}

Constructing estimators in~\eqref{eq:estimators} requires a hypothesised (likely misspecified) working dispersion structure, parametrised by $\gamma\in\Gamma$. In practice the space of working variance functions must be kept small enough to avoid overfitting and remain computationally competitive, yet large enough for a sufficiently reduced variance. As seen in Section~\ref{sec:motivating-eg} model selection criteria such as AIC(c), BIC or QIC - that choose the model which minimises an estimator of a penalised KL divergence - can perform poorly under model misspecification of the dispersion structure. 
The asymptotic variance estimator $\hat{V}$ of Algorithm~\ref{alg:V-hat} could be used as a model selection criterion that is robust to variance model misspecification. To compare $M$ working (co)variance models $\left\{\mathcal{M}_1,\ldots,\mathcal{M}_M\right\}$ one could merely calculate $\hat{V}(\mathcal{M}_1),\ldots,\hat{V}(\mathcal{M}_M)$ for each model as in Algorithm~\ref{alg:V-hat} and select the optimal dispersion model $\hat{\mathcal{M}}:=\argmin_{\mathcal{M}\in\{\mathcal{M}_1,\ldots,\mathcal{M}_M\}}c^T\hat{V}(\mathcal{M})c$. 
Whilst for a large number of models this could be relatively computationally complex, for nested models the computational burden could be alleviated by warm-starting the minimisations of of the augmented $\hat{\gamma}$ by the minimising dispersion in the largest nested model.

\begin{remark}
If post selection inference is desired, the estimator $\hat{V}(\mathcal{M}_{\hat{m}})$ no longer represents the variance of the corresponding $\beta$ estimator (as the uncertainty in $\hat{m}$ remains unaccounted for), as is standard in post selection inference. A simple (but computationally more involved) Jackknife adaptation of Algorithm~\ref{alg:V-hat} could be used in such settings for post selection inference.
\end{remark}

\section{Numerical experiments}\label{sec:numericals}

We explore the empirical properties of sandwich regression on a number of simulated and real-world datasets. Sandwich regression is compared with generalized estimating equations GEE1 (GEE)~\citep{zeger-liang} fit using the \texttt{geepack} package~\citep{geepack-code}, as well as the (extended) quasi maximum likelihood approach (EQML)~\citep{eqml, hall-qMLE-cluster} fit with the \texttt{nlme} package~\citep{nlme-code}.\footnote{In two cases we must lightly adapt the code used in \texttt{nlme} and \texttt{geepack} for our purposes: \texttt{nlme} does not explicitly allow fitting a bimomial glm with dispersion parameters estimated by EQML (Section~\ref{sec:sim-glm}); and \texttt{geepack} does not directly allow for a working covariance structure implied by a mixed effects model (Section~\ref{sec:mercado}). For details see~\url{https://github.com/elliot-young/parametric.sand.reg_simulations}.} 
Section~\ref{sec:sims} explores linear and binomial multilevel models, containing well and misspecified covariance examples, and Section~\ref{sec:real} explores two real data analyses. 

\subsection{Simulations}\label{sec:sims}
For all simulations, we compare sandwich regression with the GEE and EQML estimators as outlined above. We also study an estimator that obtains the dispersion parameter through minimising the large sample empirical sandwich loss as studied by~\citet{sandboost}. 
The unweighted estimator, obtained by fixing the working correlation in~\eqref{eq:estimators} to $\mathrm{P}_i=I_{n_i}$. 
We study a range of well and misspecified scenarios with varying sample size. In all simulations $\beta\in\R$ denotes the scalar parameter of primary interest, for which we wish to construct an estimator of minimal variance. In each case we consider $I$ i.i.d.~realisations of the respective models and vary the number of clusters $I\in[5,200]$.

\subsubsection{The multilevel linear model}\label{sec:sim-lin-mem}

Consider the linear model
\begin{gather*}
    Y_i\given X_i \sim N_4(X_i\beta,\,\Sigma(X_i)),
    \qquad
    X_i \sim N_4( {\bf 0} , I_4 ),
    \\
    \Sigma(X_i)_{jk} := \rho(j,k)\sigma(X_{ij})\sigma(X_{ik}),
    \quad
    \sigma(x) := 1 + \lambda \exp(-2x^2),
    \quad
    \rho(j,k) := \ind_{(j=k)} + 0.5\cdot\ind_{(j\neq k)},
\end{gather*}
where $\beta=1$. 
Note therefore that the case $\lambda=0$ corresponds to the intercept only mixed effects model. We analyse two settings: $\lambda=0$ (homoscedastic), and $\lambda=3$ (heteroscedastic). The class of estimators in each case is taken as~\eqref{eq:estimators} with weights the inverse of the working covariance that arises from the homoscedastic intercept only mixed effects model, and as such the former case is a model well specified setting whereas the latter is model misspecified.

\subsubsection{The multilevel binomial model}\label{sec:sim-glm}

Consider the binomial model satisfying
\begin{gather*}
    Y_{ij}\given X_i \sim \text{Bernoulli}\big(\text{expit}(X_{ij}\beta)\big),
    \qquad
    X_i \iid N_{20}( {\bf 0} \,,\, 0.5\cdot{\bf 1}{\bf 1}^T + 0.5\cdot I_{20} )
    ,
\end{gather*}
where $\beta=1$ and $\text{expit}(\eta):=\exp(\eta)/(1+\exp(\eta))$. We study two models for the intra-cluster correlation structure:

\medskip

\noindent{\bf Well specified covariance structure: } We consider an equicorrelated setting where
\begin{equation*}
    \Corr(Y_{ij}, Y_{ik}) \approx 0.4,
\end{equation*}
for $j\neq k$, 
generated through the copula
\begin{equation*}
    Z_i \sim 
    N_{20}( {\bf 0} , 0.6\cdot{\bf 1}{\bf 1}^T+0.4\cdot I_{20} ),
    \qquad
    U_{ij} := \Phi(Z_{ij}),
    \qquad
    Y_{ij} := \ind(\{ U_{ij} \geq 1-\text{expit}(X_{ij}\beta) \}) ,
\end{equation*}
where $\Phi$ is the cumulative distribution function of the standard normal distribution and is applied component-wise to $Z_i$. 
Modeling the covariance with equicorrelated structure, all estimation methods are well specified.

\medskip

\noindent{\bf Misspecified model: } Consider the data generation mechanism given by the copula
\begin{gather*}
    Z_i \iid \text{ARMA}(2,2),
    \qquad
    U_{ij} := \Phi(Z_{ij}),
    \qquad
    Y_{ij} := \ind( \{U_{ij} \geq 1-\text{expit}(X_{ij}\beta)\} ) ,
\end{gather*}
where $Z_i$ is multivariate normal with covariance matrix consisting of the first 20 terms of an $\text{ARMA}(2,2)$ model with autoregressive and moving average parameters $(0.4,0.5)$ and $(-0.9,0.4)$ 
respectively. 
The resulting correlation structure is therefore $\Corr(Y_{ij},Y_{ik}) = \rho_{|j-k|}$ with 
\begin{multline*}
    \rho = ( 1.00 ,\, -0.07 ,\, 0.41 ,\, 0.12 ,\, 0.25 ,\, 0.16 ,\, 0.19 ,\, 0.16 ,\, 0.16 ,\, 0.14 ,\, 0.14 ,\,
    \\
    0.13 ,\, 0.12 ,\, 0.12 ,\, 0.11 ,\, 0.11 ,\, 0.10 ,\, 0.10 ,\, 0.09 ,\, 0.09 ).
\end{multline*}
We fit a multilevel binomial generalized linear model to this longitudinal dataset with an $\text{AR}(1)$ working correlation, thus falls within a correlation misspecified setting.

\bigskip
The mean squared errors of all estimators are given in Figure~\ref{fig:mse_sims}. As expected, we see comparable performance at moderate sample sizes when the covariance model is well specified, but substantially improved performance of sandwich regression amongst the misspecified settings. In fact in the binomial misspecified model example both GEE and EQML methods see larger mean squared error over the baseline unweighted estimator. 
We also observe the benefit at small samples of our sandwich regression approach that utilises the finite sample sandwich loss, over the asymptotically justified alternative of the large sample sandwich loss.

\begin{figure}[ht]
    \centering
    \includegraphics[width=0.9\textwidth]{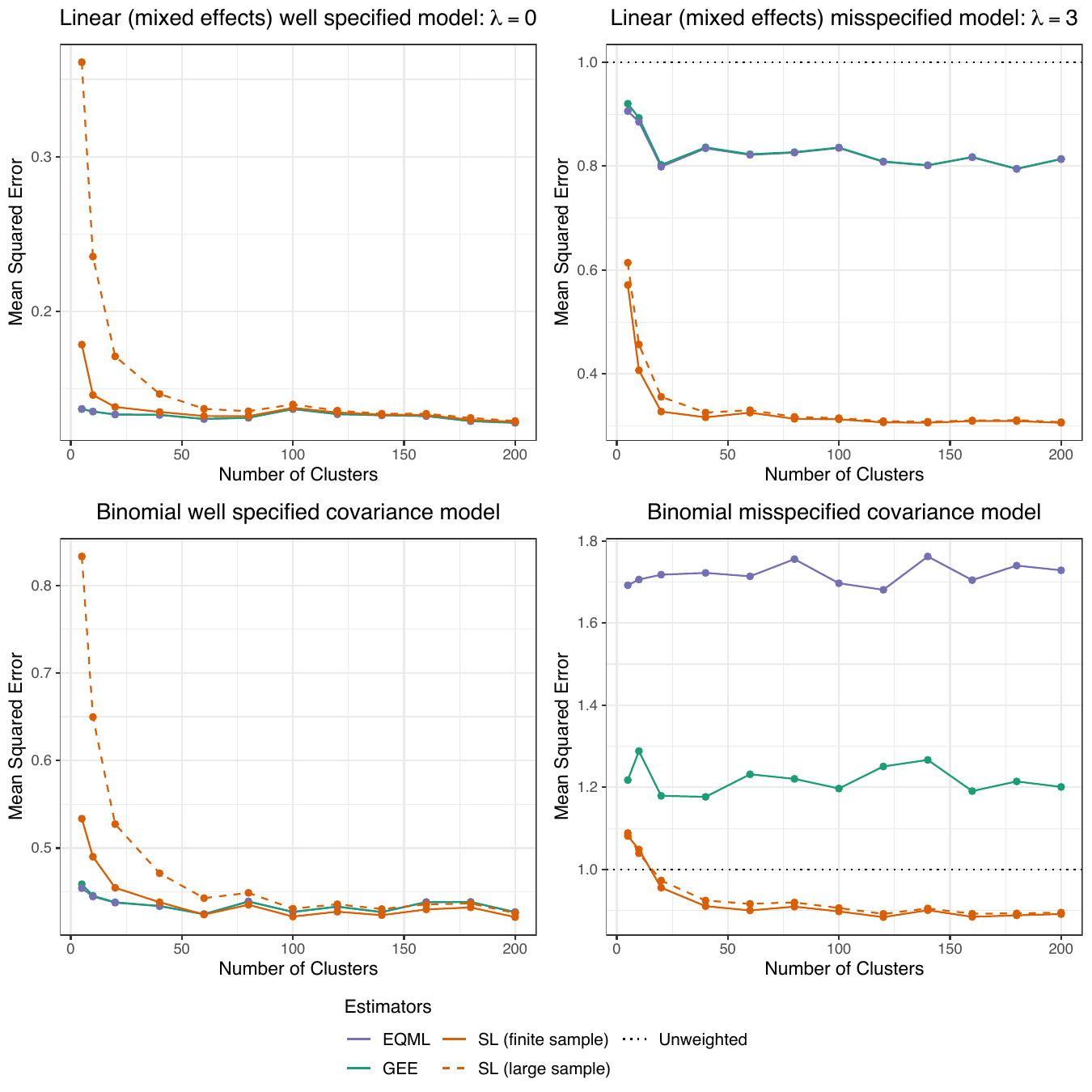}
    \caption{Mean squared error of $\beta$ estimators (relative to the unweighted estimator) in the simulations of Section~\ref{sec:sims}.}
    \label{fig:mse_sims}
\end{figure}

\subsection{Real world examples}\label{sec:real}

We study two examples to demonstrate the benefits of sandwich regression on real world datasets. The first uses a linear mixed effects model for improved inference on a causal quantity of interest in a randomised control study. The second explores a binomial glm for improved estimation of pointwise prediction estimates of regression functions.

\subsubsection{Effect of competition on market price - Linear mixed effects model example}\label{sec:mercado}

We consider a study from \citet{mercados}, which studies a randomised control trial analysing the effect of competition on prices in 72 mercados in the Dominican Republic (51 of which have at least one new retail firm added to the mercado and 21 of which do not). The dataset consists of store-level data measuring deemed prices before and after treatment, in total consisting of 1,926 observations across 72 mercados (with each mercado containing between 20 and 55 stores). This dataset has also recently been analysed by \citet{grouped-leverage}. In both analyses the following linear model is employed
\begin{equation*}
    \log\big(\text{deemed price}\big)_{ij} = \beta \cdot \ind\bigg(\parbox{11em}{\centering at least one new store\\introduced to mercado $i$}\bigg)  +  X_{ij}\alpha  +  \varepsilon_{ij},
\end{equation*}
where $(i,j)$ indexes the $j$th store in the $i$th mercado, $X_{ij}\in\R^{1\times 18}$ contains control variables of: an intercept term, the lagged deemed price (before treatment), a lagged quality index, the pre-treatment number of retailers in each district, eight provincial fixed effects, total district beneficiaries of a conditional cash transfer program, average market income, two market education measures, and a binary urban status for the market. Here $\varepsilon_{ij}$ denotes a conditional mean zero error term. 
The target parameter of interest is $\beta$; the effect of competition on deemed price. Note that both studies consider an (unweighted) clustered OLS estimate for $\beta$, with \citet{mercados} in a homoscedastic Gaussian model and \citet{grouped-leverage} in the (heteroscedastic) linear model, the latter introducing `cluster-robust' estimates for the variance of the OLS estimator. 

To construct a lower variance estimate of $\beta$ we consider the class of QML estimators for $\beta$ implied by a mixed effects model on $\varepsilon_{ij}$, specifically
\begin{equation}\label{eq:mercado-randomeffects}
    \varepsilon_{ij} = Z_{ij}u + \epsilon_{ij},
    \qquad
    Z_{ij} = \left( 1,z_{ij},z_{ij}^2 \right),
    \qquad
    z_{ij} = \Big(\parbox{6em}{\centering lagged log\\deemed price}\Big)_{ij}
\end{equation}
where $u\sim N_4({\bf 0},\mathcal{V}_u)$ are the random effects with positive definite covariance matrix $\mathcal{V}_u$, and $\epsilon_{ij}\iid N(0,\sigma^2)$ for some $\sigma>0$, and with $u\independent\epsilon_{ij}$. 

Table~\ref{tab:mercado} presents the $\beta$ estimates for each of an unweighted (OLS), mixed effects model, generalized estimating equation (GEE) and sandwich regression estimator, the latter three all selected within the class of estimators~\eqref{eq:estimators} with weights following the inverse  `working covariance' given by~\eqref{eq:mercado-randomeffects} i.e.~the only difference in estimation being the loss function used to estimate the dispersion parameters~$(\mathcal{V}_u,\sigma)$. Here the mixed effects model estimator effectively mimics the EQML estimator of~\eqref{eq:EQML-GEE}, albeit with an restricted likelihood loss. 
Also included in Table~\ref{tab:mercado} are estimates $\hat{V}$ for the variance of each estimator, all of which are calculated using the Jackknife estimator of the variance; as a result all these estimators are robust to misspecification of the random effects structure used (in contrast to e.g.~the estimators of variance typically outputted by mixed effects models packages such as \texttt{lme4}, which are not robust to such misspecification). 
We perform 10 Newton-type iterations initialised at the full sample loss to approximate the leave-one-out dispersion parameter estimates in Algorithm~\ref{alg:V-hat}. 
As we see from Table~\ref{tab:mercado} the sandwich regression estimator outperforms all other estimators in terms of variance, with the GEE and mixed effects estimators providing little to no improvements over the unweighted OLS estimator. 
Note the 95\% confidence intervals for $\beta$ constructed with each estimator overlap (in fact all their 20\% confidence intervals overlap). 

\begin{table}[ht]
	\begin{center}
  \begin{tabular}{ccccc}
  \toprule
	       Method & {\begin{tabular}{@{}c@{}c@{}}$\hat{\beta}$\\ $(\times 10^{-2})$ \end{tabular}} & {\begin{tabular}{@{}c@{}c@{}}$\hat{V}$\\$(\times 10^{-5})$ \end{tabular}} & {\begin{tabular}{@{}c@{}c@{}}Reduction in $\hat{V}$ relative to\\  unweighted (OLS) estimator (\%)\end{tabular}}
        \\
        \midrule
        Unweighted (OLS) & -1.47 & 8.29 & 0\%
        \\
        Sandwich Regression & -1.06 & {\bf 6.76} & {\bf 18.4\%}
        \\
        Mixed Effects Model & -1.25 & 7.91 &  4.6\%
        \\
        GEE & -1.45 & 8.51 & -2.7\%
        \\
        \bottomrule
        \end{tabular}
	\caption{Estimates of the effect of competition on price in the mercado pricing dataset of \citet{mercados}.}\label{tab:mercado}
	\end{center}
\end{table}

\subsubsection{1988 Bangladesh Fertility Survey - Binomial generalized linear model example}\label{sec:bang}

We study the 1988 Bangladesh Fertility Survey \citep{bang2}, which comprises of 1934 women in 60 districts. 
We fit the binomial model
\begin{gather*}
    \ind\bigg(\parbox{10em}{\centering woman $j$ in district $i$\\using contraception}\bigg) \sim \text{Bernoulli}(p_{ij}),
    \\
    g(p_{ij}) = \beta_1
    + \beta_2 \ind\bigg(\parbox{2.5em}{\centering 1 \\ child}\bigg)_{ij}
    + \beta_3 \ind\bigg(\parbox{3.5em}{\centering 2 \\ children}\bigg)_{ij}
    + \beta_4 \ind\bigg(\parbox{4em}{\centering at least 2 \\ children}\bigg)_{ij}
    + \beta_5 \ind\bigg(\parbox{3em}{\centering urban \\ area}\bigg)_{ij}
    \\
    \hspace{10em}
    + \beta_6 \big(\text{age}\big)_{ij}
    + \beta_7 \big(\text{age}\big)^2_{ij},
\end{gather*}
where $g(p) = \text{logit}(p) := \log(p/(1-p))$ is the canonical link function. Here we consider a function as the target of interest, specifically the probability that a woman with one child in an urban area of a given age is using contraception i.e.~$\text{age}\mapsto g^{-1}\big((1,1,0,0,1,\text{age},\text{(age)}^2)^T\beta\big)$. 
For all estimators a single parameter equicorrelated working correlation structure is imposed. Figure~\ref{fig:real.data.glm} plots the estimators for the pointwise asymptotic variances of $g^{-1}\big((1,1,0,0,1,\text{age},\text{(age)}^2)^T\hat{\beta}\big)$, obtained via the delta method. 
We see that sandwich regression outperforms all the competing estimators over all ages, with improvements more pronounced at some ages than others. 

\begin{figure}[ht]
    \centering
    \includegraphics[width=0.7\textwidth]{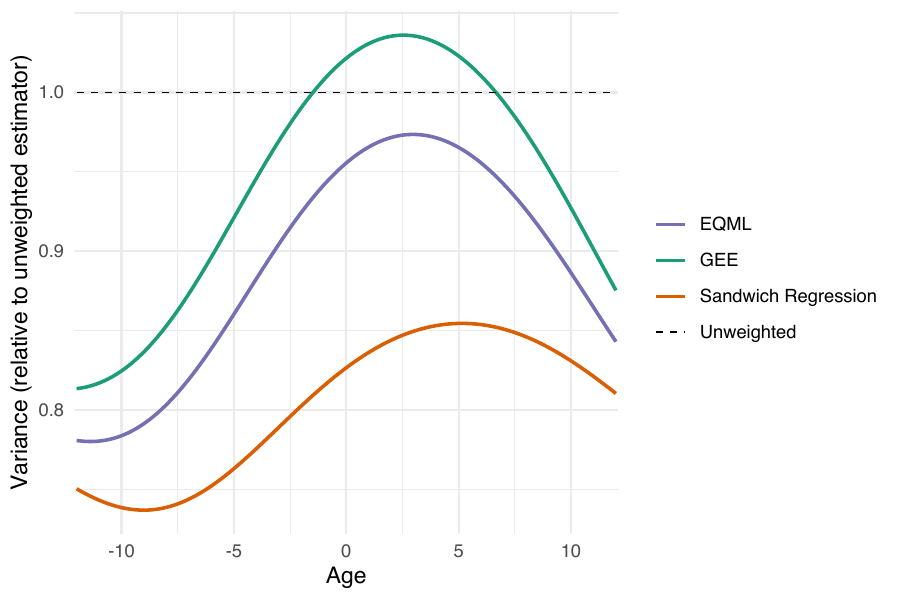}
    \caption{Estimates of  $\Var\,g^{-1}\big((1,1,0,0,1,\text{age},\text{(age)}^2)^T\hat{\beta}\big)$ for different ages.
    }
    \label{fig:real.data.glm}
\end{figure}

\section{Discussion}
It is commonly said model misspecification should be recognised in estimation tasks~\citep{box, assumption-lean}. In settings where a practitioner only has access to a moderate quantity of data it may be unreasonable to make the minimal possible identifiability assumptions on the data generating distribution, and some form of model structure may be relied on. We have studied a broad class of (multilevel) generalized linear models where the first conditional moment of the response takes a parametric form, but no additional assumptions are made on the conditional covariance. 
We recover a notion of optimality within a user-chosen restricted class of estimators by minimising an empirical estimator of the sandwich loss, that is seen to work effectively on a number of simulated and real-world datasets.

The idea of relinquishing efficiency in light of computational or robustness considerations is certainly not new. 
Asking for a sense of global optimality given constraints on the estimator class has been studied for certain shape constrained classes: \citet{hampel} considers minimal variance estimation given an upper bound on the gross error sensitivity; and \citet{feng} consider convexity constraints on the loss function used in linear regression. 
Broadly speaking, this work aims to add to a growing body of work that maintains notions of optimality globally within a broad class of models given constraints on the class of estimators considered. 
Given the increasingly complex data structures studied in practice, and correspondingly the potential necessity for constrained estimation, we look forward to seeing future developments and insights more broadly in this field.


\newpage

\appendix

{
\Large\bf
\begin{center}
Supplementary material for ‘Sandwich regression for accurate and robust estimation in
generalized linear multilevel and longitudinal models’, by Elliot H. Young and Rajen D. Shah
\end{center}
}

\section{Proofs of results}

\begin{proof}[Proof of Proposition~\ref{prop:divratio}]
    We define a family of distributions $P_\delta$ for $(Y,X)$ in terms of the constant $\delta > 0$, all of which satisfy the conditions~\eqref{eq:conditions}; to obtain the lower bound~\eqref{eq:div-ratio} we will take $\delta$ sufficiently large.
    
    We first generate a class of marginal distributions on $X$ parametrised by $\delta>0$. Define the constant
    $$\lambda_2 := \frac{1}{2}\vee\bigg(1-\frac{\tau^2}{2}\bigg),$$ noting $\lambda_2\in\big[\frac{1}{2},1\big)$. 
    Also, given $\delta>0$ define the functions
    \begin{gather*}
        \lambda_1(\delta) := \frac{1-\lambda_2}{2\int_0^\delta\frac{1}{1+x^4}dx},
        \qquad
        \nu^2(\delta) := \frac{\tau^2}{\lambda_2} - \frac{1-\lambda_2}{\lambda_2}\rho(\delta),
        \qquad
        \rho(\delta) := \frac{\int_0^\delta\frac{x^2}{1+x^4}dx}{\int_0^\delta\frac{1}{1+x^4}dx}.
    \end{gather*}
    It can be verified that $\rho(\delta)\in(0,1)$ for all $\delta>0$ (as $\rho$ can be shown to be a strictly increasing function with $\lim_{\delta\downarrow0}\rho(\delta)=0$ and $\lim_{\delta\to\infty}\rho(\delta)=1$) and thus
    $$\inf_{\delta>0}\nu^2(\delta) \geq \frac{\tau^2-1+\lambda_2}{\lambda_2} = 1-\frac{2(1-\tau^2)}{1\vee(2-\tau^2)} > 0,$$
    noting $\tau>0$. 
    Now take $X$ to have the density
    \begin{equation}\label{eq:X}
        p_\delta(x) = \frac{\lambda_1(\delta)}{1+x^4}\ind_{\{|x|\leq\delta\}} + \frac{\lambda_2}{\sqrt{2\pi\nu^2(\delta)}}e^{-\frac{x^2}{2\nu^2(\delta)}}.
    \end{equation}
    Note in particular that $p_\delta$ is positive on $\R$ with
    \begin{align*}
        \int_{\R} p_\delta(x)dx &= \lambda_1(\delta)\int_{-\delta}^\delta\frac{1}{1+x^4}dx + \lambda_2 = 1-\lambda_2+\lambda_2=1,
        \\
        \int_{\R} x p_\delta(x)dx &= 0,
        \\
        \int_{\R} x^2p_\delta(x)dx &= 2\lambda_1(\delta)\int_0^\delta\frac{x^2}{1+x^4}dx + \lambda_2\nu^2(\delta) = (1-\lambda_2)\lambda_1(\delta) + \tau^2 - (1-\lambda_2)\lambda_1(\delta) = \tau^2.
    \end{align*}
    Thus for any $\delta>0$ if $X$ follows the distribution with density $p_\delta$ then $\int_{\R} x^2 p_\delta(x)dx = \tau^2$. 
    Further, 
    \begin{equation*}
        \int_{\R} x^4p_\delta(x)dx 
        = 2\lambda_1(\delta)\int_0^\delta\frac{x^4}{1+x^4}dx
        = (1-\lambda_2) B(\delta),
    \end{equation*}
    where
    \begin{equation}\label{eq:B(delta)}
        B(\delta) := \frac{\int_0^\delta\frac{x^4}{1+x^4}dx}{\int_0^\delta\frac{1}{1+x^4}dx}.
    \end{equation}
    
    We now define a conditional law for $Y\given X$. Take
     \begin{equation}\label{eq:Y|X}
        Y\given X \sim N\big( X\beta, \,\sigma^2(X)\big),
        \qquad
        \sigma^2(X) := c_1 + c_2 X^2\ind_{[0,\infty)}(X),
    \end{equation}
    for a pair of positive constants $(c_1,c_2)$ satisfying $c_1+\frac{c_2\tau^2}{2}=\sigma^2$ and $c_2\leq c_1$, to be determined shortly. 
    We denote by $P_\delta$ the joint distribution on $(Y,X)$ admitting the density $p(y\given x) p_\delta(x)$ where $p(y\given x)$ is the density given by~\eqref{eq:Y|X} and $p_\delta(x)$ is the density~\eqref{eq:X}. 
    We also introduce the law $P'\in\Hom$ to be the distribution with conditional distribution $Y\given X\sim N(X\beta, c_1)$, for which we denote the corresponding conditional density $p'(y\given x)$, and with marginal density for $X$ given by $p_\delta(x)$. Then 
    \begin{align*}
        \inf_{\tilde{P}\in\Hom}\mathrm{KL}(\tilde{P}\,\|\,P_\delta)
        &=
        \mathrm{KL}(P'\,\|\,P_\delta)
        \\
        &=
        \int_{\R^2} p'(y\given x)p_\delta(x)\log\bigg(\frac{p'(y\given x)p_\delta(x)}{p(y\given x)p_\delta(x)}\bigg)\,dydx
        \\
        &=
        \int_{\R} p_\delta(x) \int_{\R} p'(y\given x)\log\bigg(\frac{p'(y\given x)}{p(y\given x)}\bigg)\,dydx
        \\
        &=
        \frac{1}{2}\int_{\R} p_\delta(x)\Big\{\log\big(1+c_1^{-1}c_2x^2\ind_{[0,\infty)}(x)\big) + \frac{1}{1+c_1^{-1}c_2x^2\ind_{[0,\infty)}(x)}-1\Big\}\,dx
        \\
        &= \frac{1}{2}\int_0^\infty p_\delta(x)\Big\{\log\big(1+c_1^{-1}c_2x^2\big) + \frac{1}{1+c_1^{-1}c_2x^2} - 1 \Big\}\,dx
        \\
        & = I\big(c_1^{-1}c_2, \delta\big),
    \end{align*}
    where
    \begin{equation*}
        I(c,\delta) := \frac{1}{2}\int_0^\infty p_\delta(x) \Big\{\log\big(1+cx^2\big) + \frac{1}{1+cx^2} - 1 \Big\}\,dx,
    \end{equation*}
    and where we use the well known result
    \begin{equation*}
        \mathrm{KL}\big(N(\mu,\sigma_1^2)\,\|\,N(\mu,\sigma_2^2)\big) = \frac{1}{2}\bigg\{\log\bigg(\frac{\sigma_2^2}{\sigma_1^2}\bigg) + \frac{\sigma_1^2}{\sigma_2^2} - 1 \bigg\},
    \end{equation*}
    in addition to using Fubini's Theorem, which can be applied as
    \begin{align*}
        &\quad\;
        \sup_{\substack{\delta,c_1,c_2\in\R_+\\c_2\leq c_1}}\int_{\R}\int_{\R}\bigg|p'(y\given x)p_\delta(x)\log\bigg(\frac{p'(y\given x)}{p(y\given x)}\bigg)\bigg|\,dydx
        \\
        &\leq \sup_{\substack{\delta,c_1,c_2\in\R_+\\c_2\leq c_1}} \frac{1}{2}\int_{\R}p_\delta(x)\int_{\R}p'(y\given x)\bigg|\log\bigg(\frac{c_1+c_2x^2\ind_{[0,\infty)}(x)}{c_1}\bigg)
        \\
        &\qquad\qquad
        +\bigg(\frac{1}{c_1+c_2x^2\ind_{[0,\infty)}(x)}-\frac{1}{c_1}\bigg)(y-x\beta)^2\bigg|\,dydx
        \\
        &\leq \sup_{\substack{\delta,c_1,c_2\in\R_+\\c_2\leq c_1}} \frac{1}{2}\int_{\R}p_\delta(x)\bigg|\log\big(1+c_1^{-1}c_2x^2\ind_{[0,\infty)}(x)\big)
        +\frac{1}{1+c_1^{-1}c_2x^2\ind_{[0,\infty)}(x)}-1\bigg|\,dx
        \\
        &\leq
        \sup_{\substack{\delta,c_1,c_2\in\R_+\\c_2\leq c_1}}\frac{1}{2}\int_{\R} p_\delta(x)\bigg( \underbrace{\log\big(1+c_1^{-1}c_2x^2\ind_{[0,\infty)}(x)\big)}_{\leq\log(1+x^2)\leq x^2} + \underbrace{\bigg|\frac{1}{1+c_1^{-1}c_2x^2\ind_{[0,\infty)}(x)}-1 \bigg|}_{\leq1}\bigg) dx
        \\
        &\leq \frac{\tau^2+1}{2} < \infty.
    \end{align*}
    The integral function $I(c,\delta)$ 
    is well-defined on  $(0,1]\times\R_+$, and can also easily be shown to be differentiable (its derivatives take explicit forms in terms of integrals that are also well-defined). Thus $I(c,\delta)$ is continuous on $(0,1]\times\R_+$. Further 
    \begin{multline*}
        \sup_{c\in(0,1]}\sup_{\delta>0} I(c,\delta) 
        = 
        \frac{1}{2}\sup_{c\in(0,1]}\sup_{\delta>0}\int_0^\infty p_\delta(x)\Big|\log\big(1+cx^2\big)+\frac{1}{1+cx^2}-1\Big|
        \\
        \leq
        \frac{1}{2}\sup_{c\in(0,1]}\sup_{\delta>0}\int_0^\infty p_\delta(x)\Big\{\underbrace{\log\big(1+cx^2\big)}_{\leq\log(1+x^2)\leq x^2}+\underbrace{\Big|\frac{1}{1+cx^2}-1\Big|}_{\leq1}\Big\}
        \\
        \leq
        \frac{1}{2}\sup_{c\in(0,1]}\sup_{\delta>0}\Big(\int_0^\infty x^2p_\delta(x)dx + \int_0^\infty p_\delta(x)dx\Big)
        \leq
        \frac{\tau^2+1}{4}
        ,
    \end{multline*}
    is bounded. 
    Also as $I(0,\delta)=0$ for all $\delta>0$ it follows that $i(c):=\sup_{\delta>0}I(c,\delta)$ is continuous on $(0,1]$ with $i(0)=0$. Therefore for any $\alpha>0$ there exists a $\tilde{c}\in(0,1]$ such that $i(\tilde{c})\leq\alpha$, and thus
    \begin{equation*}
        \sup_{\delta>0}I(\tilde{c},\delta) \leq \alpha.
    \end{equation*}
    Now take 
    \begin{equation}\label{eq:c1c2}
        c_1 := \frac{2\sigma^2}{2+\tilde{c}\tau^2}
        ,\qquad
        c_2 := \frac{2\tilde{c}\sigma^2}{2+\tilde{c}\tau^2}.
    \end{equation}
    It then follows that
    \begin{equation*}
        \inf_{\tilde{P}\in\Hom}\mathrm{KL}(\tilde{P}\,\|\,P_\delta)
        \leq I(c_1^{-1}c_2,\delta) \leq \sup_{\delta>0}I(\tilde{c},\delta) \leq \alpha,
    \end{equation*}
    for all $\delta>0$.  
    Additionally, by construction of $(c_1,c_2)$ in~\eqref{eq:c1c2},
    \begin{equation*}
        \E_{P_\delta}[\Var_{P_\delta}(Y\given X)] = c_1 + c_2\,\E_{P_\delta}[X^2\ind_{[0,\infty)}(X)]
        =
        c_1 + \frac{c_2\tau^2}{2} = \sigma^2.
    \end{equation*}
    Note similar arguments could be employed to construct a $(c_1,c_2)$ pair such that additionally for all $\delta>0$ we have $\inf_{\tilde{P}\in\Hom}\mathrm{KL}(P_\delta\,\|\,\tilde{P})\leq\alpha$.

    To summarise, to this point we have constructed a class of laws $P_\delta$ that satisfy
    \begin{equation*}
        \inf_{\tilde{P}\in\Hom}\mathrm{KL}(\tilde{P}\,\|\,P_\delta),
        \quad
        \Var_{P_\delta}(X)=\tau^2,
        \quad
        \E_{P_\delta}[\Var_{P_\delta}(Y\given X)]=\sigma^2,
    \end{equation*}
    for any $\delta>0$. The remainder of the proof will construct a specific law within this class such that~\eqref{eq:div-ratio} also holds. For the remainder of the proof unless stated otherwise expectations, variances and dispersion parameters $\gamma$ should be taken as with respect to $P_\delta$.

    Now, it can be shown that the minimisers of each of the EQML and GEE losses (see e.g.~\citet{sandboost}) over the class~\eqref{eq:working-var} 
    are
    \begin{align*}
        \gamma_{\QML,1} &= \gamma_{\GEE,1} = \E[\sigma^2(X)\given X\geq0] = c_1 + c_2\tau^2,
        \\
        \gamma_{\QML,2} &= \gamma_{\GEE,2} = \E[\sigma^2(X)\given X<0] = c_1,
    \end{align*}
    and thus writing $\tilde{\gamma}:=\gamma_{\QML}=\gamma_{\GEE}$ it follows that
    \begin{equation*}
        V(\gamma_{\QML}) = V(\gamma_{\GEE}) = \E\bigg[\frac{X^2}{\nu_{\tilde{\gamma}}(X)}\bigg]^{-2}\E\bigg[\frac{X^2\sigma^2(X)}{\nu^2_{\tilde{\gamma}}(X)}\bigg].
    \end{equation*}
    Further, it follows immediately from the Cauchy--Schwarz Theorem that
    \begin{equation*}
        \inf_{\gamma\in\Gamma}V(\gamma) = \E\bigg[\frac{X^2}{\sigma^2(X)}\bigg]^{-1}.
    \end{equation*}
    Thus
    \begin{align*}
        &\quad
        \frac{V(\gamma_{\QML})\wedge V(\gamma_{\GEE})}{\inf_{\gamma\in\Gamma}V(\gamma)}
        \\
        &=
        \E\bigg[\frac{X^2}{\nu_{\tilde{\gamma}}(X)}\bigg]^{-2}\E\bigg[\frac{X^2}{\sigma^2(X)}\bigg]\E\bigg[\frac{X^2\sigma^2(X)}{\nu^2_{\tilde{\gamma}}(X)}\bigg]
        \\
        &=
        \Big(\tilde{\gamma}_1^{-1}\E\big[X^2\ind_{[0,\infty)}(X)\big]+\tilde{\gamma}_2^{-1}\E\big[X^2\ind_{(-\infty,0)}(X)\big]\Big)^{-2}
        \\
        &\qquad\cdot
        \Big(\tilde{\gamma}_1^{-2}\E\big[X^2\sigma^2(X)\ind_{[0,\infty)}(X)\big]+\tilde{\gamma}_2^{-2}\E\big[X^2\sigma^2(X)\ind_{(-\infty,0)}(X)\big]\Big)
        \\
        &\qquad\cdot
        \bigg(\E\bigg[\frac{X^2}{\sigma^2(X)}\ind_{[0,\infty)}(X)\bigg]+\E\bigg[\frac{X^2}{\sigma^2(X)}\ind_{(-\infty,0)}(X)\bigg]\bigg)
        \\
        &=
        \Big((c_1+c_2\tau^2)^{-1}\E\big[X^2\ind_{[0,\infty)}(X)\big]+c_1^{-1}\E\big[X^2\ind_{(-\infty,0)}(X)\big]\Big)^{-2}
        \\
        &\qquad
        \cdot
        {\Big((c_1+c_2\tau^2)^{-2}\big(c_1\E\big[X^2\ind_{[0,\infty)}(X)\big]+c_2\E\big[X^4\ind_{[0,\infty)}(X)\big]\big)+c_1^{-1}\E\big[X^2\ind_{(-\infty,0)}(X)\big]\Big)}
        \\
        &\qquad \cdot\bigg(\E\bigg[\frac{X^2}{c_1+c_2X^2}\ind_{[0,\infty)}(X)\bigg]+c_1^{-1}\E\big[X^2\ind_{(-\infty,0)}(X)\big]\bigg)
        \\
        &\geq
        \frac{(c_1+c_2\tau^2)^{-2}c_2\E[X^4\ind_{[0,\infty)}(X)]\cdot c_1^{-1}\E[X^2\ind_{(-\infty,0)}(X)]}{\big(\frac{1}{c_1+c_2\tau^2}+\frac{1}{c_1}\big)^2\big(\frac{\tau^2}{2}\big)^2}
        \\
        &=
        \frac{c_1c_2}{2\tau^2(c_1+2c_2
        \tau^2)}\,\E[X^4\ind_{[0,\infty)}(X)]
        \\
        &=
        \frac{c_1c_2(1-\lambda_2)}{2\tau^2(c_1+2c_2
        \tau^2)}
        \,
        B(\delta),
    \end{align*}
    making use of
    \begin{equation*}
        \E[X^2\ind_{[0,\infty)}(X)] = \E[X^2\ind_{(-\infty,0)}(X)] = \int_0^\delta x^2 p_\delta(x)dx = \frac{1}{2}\E[X^2] = \frac{\tau^2}{2},
    \end{equation*}
    as the density $p_\delta$ is symmetric about zero. Recall that $\frac{2\tau^2(c_1+2c_2\tau^2)}{c_1c_2(1-\lambda_2)}>0$. Further, the function $B$ in~\eqref{eq:B(delta)} is continuous in $\delta$ with $B(0)=0$ and $B(\delta)\to+\infty$ as $\delta\to+\infty$, and so for any $\eta>0$ there exists some $\delta^*>0$ such that $B(\delta^*)
 \geq\frac{2\tau^2(c_1+2c_2\tau^2)}{c_1c_2(1-\lambda_2)}\,\eta$. Thus, the law $P_{\delta^*}$ on $(Y,X)$ satisfies
 \begin{equation*}
     \frac{V(\gamma_{\QML})\wedge V(\gamma_{\GEE})}{\inf_{\gamma\in\Gamma}V(\gamma)} \geq \eta,
 \end{equation*}
 alongside the additional conditions~\eqref{eq:conditions} of  Proposition~\ref{prop:divratio} (that were shown to hold for all $P_\delta$ with $\delta>0$ and thus hold for $P_{\delta^*}$).
\end{proof}


\begin{thebibliography}{}

\bibitem[Beirami et~al., 2017]{aloocv}
Beirami, A., Razaviyayn, M., Shahrampour, S., and Tarokh, V. (2017).
\newblock On optimal generalizability in parametric learning.
\newblock {\em Advances in Neural Information Processing Systems (NeurIPS)},
  pages 3455--3465.

\bibitem[Box et~al., 1994]{box}
Box, G. E.~P., Jenkins, G.~M., Reinsel, G.~C., and Ljung, G.~M. (1994).
\newblock {\em Time series analysis : forecasting and control.}
\newblock Prentice Hall, Englewood Cliffs, NJ, 5 edition.

\bibitem[Brockwell and Davis, 1991]{Brockwell}
Brockwell, P.~J. and Davis, R.~A. (1991).
\newblock {\em Time Series: Theory and Methods}.
\newblock Springer Series in Statistics. Springer New York, 2 edition.

\bibitem[Busso and Galiani, 2019]{mercados}
Busso, M. and Galiani, S. (2019).
\newblock The causal effect of competition on prices and quality: Evidence from
  a field experiment.
\newblock {\em American Economic Journal: Applied Economics}, 11(1):33--56.

\bibitem[Cribari-Neto and da~Silva, 2011]{HC4-5}
Cribari-Neto, F. and da~Silva, W.~B. (2011).
\newblock A new heteroskedasticity-consistent covariance matrix estimator for
  the linear regression model.
\newblock {\em AStA Advances in Statistical Analysis}, 95(2):129--146.

\bibitem[Diggle et~al., 2013]{glmmbook}
Diggle, P.~J., Heagerty, P., Liang, K.-Y., and Zeger, S.~L. (2013).
\newblock {\em Analysis of longitudinal data}.
\newblock Oxford University Press, 2013., 2 edition.

\bibitem[Efron and Stein, 1981]{efron}
Efron, B. and Stein, C. (1981).
\newblock The jackknife estimate of variance.
\newblock {\em The Annals of Statistics}, 9(3):586 -- 596.

\bibitem[Engle, 1982]{garch}
Engle, R.~F. (1982).
\newblock Autoregressive conditional heteroscedasticity with estimates of the
  variance of united kingdom inflation.
\newblock {\em Econometrica}, 50(4):987--1007.

\bibitem[Feng et~al., 2024]{feng}
Feng, O.~Y., Kao, Y.-C., Xu, M., and Samworth, R.~J. (2024).
\newblock Optimal convex $m$-estimation via score matching.
\newblock {\em arXiv preprint arXiv:2403.16688}.

\bibitem[Gourieroux and Monfort, 1993]{gourieroux}
Gourieroux, C. and Monfort, A. (1993).
\newblock Pseudo-likelihood methods.
\newblock In {\em Econometrics}, volume~11 of {\em Handbook of Statistics},
  pages 335--362. Elsevier.

\bibitem[Gourieroux et~al., 1984]{gourieroux2}
Gourieroux, C., Monfort, A., and Trognon, A. (1984).
\newblock Pseudo maximum likelihood methods: Theory.
\newblock {\em Econometrica}, 52(3):681--700.

\bibitem[Halekoh et~al., 2006]{geepack-code}
Halekoh, U., H{\o}jsgaard, S., and Yan, J. (2006).
\newblock The r package geepack for generalized estimating equations.
\newblock {\em Journal of Statistical Software}, 15/2:1--11.

\bibitem[Hall, 2001]{hall-qMLE-cluster}
Hall, D.~B. (2001).
\newblock On the application of extended quasi-likelihood to the clustered data
  case.
\newblock {\em The Canadian Journal of Statistics}, 29(1):77--97.

\bibitem[Hall and Severini, 1998]{severini1}
Hall, D.~B. and Severini, T.~A. (1998).
\newblock Extended generalized estimating equations for clustered data.
\newblock {\em Journal of the American Statistical Association},
  93(444):1365--1375.

\bibitem[Hampel, 1974]{hampel}
Hampel, F.~R. (1974).
\newblock The influence curve and its role in robust estimation.
\newblock {\em Journal of the American Statistical Association},
  69(346):383--393.

\bibitem[Hardin and Hilbe, 2003]{hardin}
Hardin, J.~W. and Hilbe, J.~M. (2003).
\newblock {\em Generalized estimating equations}.
\newblock Chapman and Hall.

\bibitem[Holst and J{\o}rgensen, 2010]{holst}
Holst, R. and J{\o}rgensen, B. (2010).
\newblock Efficient and robust estimation for a class of generalized linear
  longitudinal mixed models.
\newblock {\em arXiv preprint arXiv:1008.2870}.

\bibitem[Huber, 1967]{huber}
Huber, P.~J. (1967).
\newblock The behaviour of maximum likelihood estimates under nonstandard
  conditions.
\newblock {\em Proceedings of the Fifth Berkeley Symposium}, pages 221--223.

\bibitem[Huq and Cleland, 1990]{bang2}
Huq, N.~M. and Cleland, J. (1990).
\newblock Bangladesh fertility survey 1989 (main report).

\bibitem[J{\o}rgensen, 2013]{medf2}
J{\o}rgensen, B. (2013).
\newblock Construction of multivariate dispersion models.
\newblock {\em Brazilian Journal of Probability and Statistics},
  27(3):285--309.

\bibitem[J{\o}rgensen and Knudsen, 2004]{jorg-knudsen}
J{\o}rgensen, B. and Knudsen, S.~J. (2004).
\newblock Parameter orthogonality and bias adjustment for estimating functions.
\newblock {\em Scandinavian Journal of Statistics}, 31(1):93--114.

\bibitem[J{\o}rgensen and Martınez, 2013]{medf}
J{\o}rgensen, B. and Martınez, J.~R. (2013).
\newblock Multivariate exponential dispersion models.
\newblock In {\em Multivariate Statistics: Theory And Applications -
  Proceedings Of The Ix Tartu Conference On Multivariate Statistics And Xx
  International Workshop On Matrices And Statistics: Theory and Applications},
  pages 73--98. World Scientific Publishing Company.

\bibitem[Lee et~al., 2006]{lee-nelder-book}
Lee, Y., Nelder, J., and Pawitan, Y. (2006).
\newblock {\em Generalized Linear Models with Random Effects: Unified Analysis
  via H-likelihood}.
\newblock Chapman \& Hall CRC Monographs on Statistics \& Applied Probability.
  CRC Press.

\bibitem[Lee and Nelder, 1996]{lee-nelder}
Lee, Y. and Nelder, J.~A. (1996).
\newblock Hierarchical generalized linear models.
\newblock {\em Journal of the Royal Statistical Society. Series B
  (Methodological)}, 58(4):619--678.

\bibitem[Liang et~al., 1992]{gee1}
Liang, K.-Y., Zeger, S.~L., and Qaqish, B. (1992).
\newblock Multivariate regression analyses for categorical data.
\newblock {\em Journal of the Royal Statistical Society. Series B
  (Methodological)}, 54(1):3--40.

\bibitem[Ma and J{\o}rgensen, 2007]{ma}
Ma, R. and J{\o}rgensen, B. (2007).
\newblock Nested generalized linear mixed models: An orthodox best linear
  unbiased predictor approach.
\newblock {\em Journal of the Royal Statistical Society. Series B (Statistical
  Methodology)}, 69(4):625--641.

\bibitem[MacKinnon et~al., 2023a]{mackinnon2023fast}
MacKinnon, J.~G., Nielsen, M.~{\O}., and Webb, M.~D. (2023a).
\newblock Fast and reliable jackknife and bootstrap methods for cluster-robust
  inference.
\newblock {\em Journal of Applied Econometrics}, 38(5):671--694.

\bibitem[MacKinnon et~al., 2023b]{grouped-leverage}
MacKinnon, J.~G., Nielsen, M.~{\O}., and Webb, M.~D. (2023b).
\newblock Leverage, influence, and the jackknife in clustered regression
  models: Reliable inference using summclust.
\newblock {\em The Stata Journal}, 23(4):942--982.

\bibitem[MacKinnon and White, 1985]{mackinnon-white}
MacKinnon, J.~G. and White, H. (1985).
\newblock Some heteroskedasticity-consistent covariance matrix estimators with
  improved finite sample properties.
\newblock {\em Journal of Econometrics}, 29(3):305--325.

\bibitem[McCullagh, 1983]{qml}
McCullagh, P. (1983).
\newblock Quasi-likelihood functions.
\newblock {\em The Annals of Statistics}, 11(1):59 -- 67.

\bibitem[McCullagh and Nelder, 1989]{glmbook}
McCullagh, P. and Nelder, J. (1989).
\newblock {\em Generalized linear models}.
\newblock Monographs on statistics and applied probability (Series) ; 37.
  Chapman and Hall, second edition. edition.

\bibitem[McCullagh and Tibshirani, 1990]{mccullagh-tib}
McCullagh, P. and Tibshirani, R. (1990).
\newblock A simple method for the adjustment of profile likelihoods.
\newblock {\em Journal of the Royal Statistical Society. Series B
  (Methodological)}, 52(2):325--344.

\bibitem[Nelder and Pregibon, 1987]{eqml}
Nelder, J.~A. and Pregibon, D. (1987).
\newblock An extended quasi-likelihood function.
\newblock {\em Biometrika}, 74(2):221--232.

\bibitem[Nelder and Wedderburn, 1972]{glms}
Nelder, J.~A. and Wedderburn, R. W.~M. (1972).
\newblock Generalized linear models.
\newblock {\em Journal of the Royal Statistical Society. Series A (General)},
  135(3):370--384.

\bibitem[Pan, 2001]{QIC}
Pan, W. (2001).
\newblock Akaike's information criterion in generalized estimating equations.
\newblock {\em Biometrics}, 57(1):120--125.

\bibitem[Pinheiro et~al., 2022]{nlme-code}
Pinheiro, J., Bates, D., and {R Core Team} (2022).
\newblock {\em nlme: Linear and Nonlinear Mixed Effects Models}.
\newblock R package version 3.1-161.

\bibitem[Pinheiro and Bates, 2000]{membook}
Pinheiro, J.~C. and Bates, D.~M. (2000).
\newblock {\em Mixed-Effects Models in S and S-PLUS}, volume~1 of {\em Springer
  Statistics and Computing}.
\newblock Springer.

\bibitem[Royall, 1986]{royall}
Royall, R.~M. (1986).
\newblock Model robust confidence intervals using maximum likelihood
  estimators.
\newblock {\em International Statistical Review / Revue Internationale de
  Statistique}, 54(2):221--226.

\bibitem[Severini, 2002]{severini2}
Severini, T.~A. (2002).
\newblock Modified estimating functions.
\newblock {\em Biometrika}, 89(2):333--343.

\bibitem[Tsiatis, 2006]{tsiatis}
Tsiatis, A.~A. (2006).
\newblock {\em Semiparametric theory and missing data}, volume~1 of {\em
  Springer series in statistics}.
\newblock Springer, New York.

\bibitem[van~der Laan and Robins, 2003]{vanderlaan-robins}
van~der Laan, M.~J. and Robins, J.~M. (2003).
\newblock {\em Unified Methods for Censored Longitudinal Data and Causality}.
\newblock Springer series in statistics. Springer New York.

\bibitem[van~der Laan and Rose, 2018]{targeted-learning}
van~der Laan, M.~J. and Rose, S. (2018).
\newblock {\em Targeted Learning in Data Science Causal Inference for Complex
  Longitudinal Studies}.
\newblock Springer series in statistics. Springer New York, NY, 1 edition.

\bibitem[Vansteelandt and Dukes, 2022]{assumption-lean}
Vansteelandt, S. and Dukes, O. (2022).
\newblock Assumption-lean inference for generalised linear model parameters.
\newblock {\em Journal of the Royal Statistical Society Series B: Statistical
  Methodology}, 84(3):657--685.

\bibitem[Wedderburn, 1974]{qmle}
Wedderburn, R. W.~M. (1974).
\newblock Quasi-likelihood functions, generalized linear models, and the
  gauss-newton method.
\newblock {\em Biometrika}, 61(3):439--447.

\bibitem[White, 1980]{white}
White, H. (1980).
\newblock A heteroskedasticity-consistent covariance matrix estimator and a
  direct test for heteroskedasticity.
\newblock {\em Econometrica}, 48(4):817--838.

\bibitem[Young and Shah, 2024]{sandboost}
Young, E.~H. and Shah, R.~D. (2024).
\newblock Sandwich boosting for accurate estimation in partially linear models
  for grouped data.
\newblock {\em Journal of the Royal Statistical Society Series B: Statistical
  Methodology}, 86(5):1286--1311.

\bibitem[Zeger and Liang, 1986]{zeger-liang}
Zeger, S.~L. and Liang, K.-Y. (1986).
\newblock Longitudinal data analysis for discrete and continuous outcomes.
\newblock {\em Biometrics}, 42(1):121--130.

\bibitem[Zhang et~al., 2008]{zhang}
Zhang, M., Tsiatis, A.~A., and Davidian, M. (2008).
\newblock Improving efficiency of inferences in randomized clinical trials
  using auxiliary covariates.
\newblock {\em Biometrics}, 64(3):707--715.

\bibitem[Ziegler, 2011]{ziegler}
Ziegler, A. (2011).
\newblock {\em Generalized estimating equations}.
\newblock Lecture notes in statistics. Springer, New York.

\end{thebibliography}
\end{document}